%% file: main.tex
\documentclass[sigconf]{acmart}

\usepackage{booktabs} 

\setcopyright{rightsretained}

\newcommand{\reals}{\mathbb{R}}
\newcommand{\flowpipe}{\mathcal{F}}

\acmConference[EMSOFT'19]{ACM conference}{October 2019}{
  New York, USA}
\acmYear{2019}
\copyrightyear{2016}

\acmArticle{4}
\acmPrice{15.00}

\setcopyright{none}
\renewcommand\footnotetextcopyrightpermission[1]{} 
\usepackage[ruled, vlined, linesnumbered]{algorithm2e}
\usepackage{amsmath,bm,mathtools}
\usepackage{amsthm}
\usepackage{amssymb}
\usepackage{graphicx}
\usepackage{subfig}
\usepackage{xcolor}
\usepackage{multirow}

\newtheorem{remark}{Remark}

\newtheorem{proposition}{Proposition}

\DeclareMathOperator{\sgn}{sgn}
\DeclareMathOperator{\diag}{diag}

\allowdisplaybreaks

\newenvironment{myitemize}{\begin{list}{$\bullet$}
{\setlength{\topsep}{1mm}
\setlength{\itemsep}{0.25mm}
\setlength{\parsep}{0.25mm}
\setlength{\itemindent}{0mm}
\setlength{\partopsep}{0mm}
\setlength{\labelwidth}{15mm}
\setlength{\leftmargin}{4mm}}}{\end{list}}

 \newcommand{\minitab}[2][l]{\begin{tabular}{#1}#2\end{tabular}} 

\begin{document}
\title{ReachNN: Reachability Analysis of Neural-Network\\ Controlled Systems}

\author{Chao Huang}
\affiliation{%
 \institution{Northwestern University}
 \city{Evanston}
 \state{Illinois}
}
\email{chao.huang@northwestern.edu}

\author{Jiameng Fan}
\affiliation{%
 \institution{Boston University}
 \city{Boston}
 \state{Massachusetts}
}
\email{jmfan@bu.edu}

\author{Wenchao Li}
\affiliation{%
 \institution{Boston University}
 \city{Boston}
 \state{Massachusetts}
}
\email{wenchao@bu.edu}

\author{Xin Chen}
\affiliation{%
 \institution{Dayton University}
 \city{Dayton}
 \state{Ohio}
}
\email{xchen4@udayton.edu}

\author{Qi Zhu}
\affiliation{%
	\institution{Northwestern University}
	\city{Evanston}
	\state{Illinois}
}
\email{qzhu@northwestern.edu}

\begin{abstract}
    Applying neural networks as controllers in dynamical systems has shown great promises. However, it is critical yet challenging to verify the safety of such control systems with neural-network controllers in the loop. Previous methods for verifying neural network controlled systems are limited to a few specific activation functions. In this work, we propose a new reachability analysis approach based on Bernstein polynomials that can verify neural-network controlled systems with a more general form of activation functions, i.e., as long as they ensure that the neural networks are Lipschitz continuous. Specifically, we consider abstracting feedforward neural networks with Bernstein polynomials for a small subset of inputs. To quantify the error introduced by abstraction, we provide both theoretical error bound estimation based on the theory of Bernstein polynomials and more practical sampling based error bound estimation, following a tight Lipschitz constant estimation approach based on forward reachability analysis. Compared with previous methods, our approach addresses a much broader set of neural networks, including heterogeneous neural networks that contain multiple types of activation functions. Experiment results on a variety of benchmarks show the effectiveness of our approach.
\end{abstract}

%
%

\maketitle

\input{sections/intro}

\input{sections/problem}

\input{sections/approach}

\input{sections/error.tex}

\input{sections/experiment}

\input{sections/conclusion}

\bibliographystyle{ACM-Reference-Format}
\bibliography{./chao,./jiameng,./xin}

\appendix

\input{sections/appendix.tex}

\end{document}

%% file: sections/intro.tex
\section{Introduction} \label{sec:intro}

Data-driven control systems, especially neural-network-based controllers~\cite{mnih2015human, Lillicrap2016ContinuousCW, pan2018agile}, have recently become the subject of intense research and demonstrated great promises. 
Formally verifying the safety of these systems however still remains an open problem. 
A Neural-Network Controlled System (NNCS) is essentially a continuous system controlled by a neural network, which produces control inputs at the beginning of each control step based on the current values of the state variables and feeds them back to the continuous system. 
Reachability of continuous or hybrid dynamical systems with traditional controllers has been extensively studied in the last decades. It has been proven that reachability of most nonlinear systems is undecidable~\cite{Alur+/1995/hybrid_systems,henzinger1998s}. Recent approaches mainly focus on the overapproximation of reachable sets~\cite{dreossi2016parallelotope,lygeros1999controllers,frehse2005phaver,yang2016linear,prajna2004safety,huang2017probabilistic}. The main difficulty impeding the direct application of these approaches to NNCS is the hardness of formally characterizing or abstracting the input-output mapping of a neural network.


Some recent approaches considered the problem of computing the output range of a neural network. Given a neural network along with a set of the inputs, these methods seek to compute an interval or a box (vector of intervals) that contains the set of corresponding outputs. These techniques are partly motivated by the study of robustness~\cite{duttoutputa2018} of neural networks to adversarial examples~\cite{szegedy2013intriguing}. \citet{katz2017reluplex} propose an SMT-based approach called Reluplex by extending the simplex algorithm to handle ReLU constraints.
\citet{huang2017safety} use a refinement-by-layer technique to prove the absence or show the presence of adversarial examples around the neighborhood of a specific input. General neural networks with Lipschitz continuity are then considered by \citet{ruan2018reachability}, where the authors show that a large number of neural networks are Lipschitz continuous and the Lipschitz constant can help in estimating the output range which requires solving a global optimization problem. \citet{duttoutputa2018} propose an efficient approach using mixed integer linear programming to compute the exact interval range of a neural network with only ReLU activation functions. 

However, these existing methods cannot be directly used to analyze the reachability of dynamical systems controlled by neural networks. As the behavior of these systems is based on the interaction between the continuous dynamics and the neural-network controller, we need to not only compute the output range but also describe the input-output mapping for the controller. More precisely, we need to compute a tractable function model whose domain is the input set of the controller and its output range contains the set of the controller's outputs. We call such a function model a higher-order set, to highlight the distinction from intervals which are $0$-order sets. Computing a tractable function model from the original model can also be viewed as a form of \emph{knowledge distillation}~\cite{Hinton2015DistillingTK} from the verification perspective, as the function model should be able to produce comparable results or replicate the outputs of the target neural network on specific inputs.

There have been some recent efforts on computing higher-order sets for the controllers in NNCS. \citet{ivanov2018verisig} present a method to equivalently transform a system to a hybrid automaton by replacing a neuron in the controller with an ordinary differential equation (ODE). This method is however only applicable to differentiable neural-network controllers -- ReLU neural networks are thus excluded. \citet{Dutta_Others__2019__Reachability} use a flowpipe construction scheme to compute overapproximations for reachable set segments. A piecewise polynomial model is used to provide an approximation of the input-output mapping of the controller and an error bound on the approximation. This method is however limited to neural networks with ReLU activation functions. 
We will discuss technical features of these related works in more detail in Section~\ref{sec:problem} when we introduce the problem formally.

Neural network controllers in practical applications could involve multiple types of activation functions~\cite{beer1989heterogeneous,Lillicrap2016ContinuousCW}. The approaches discussed above for specific activation function may not be able to handle such cases, and a more general approach is thus needed.

In this paper, we propose a new reachability analysis approach for verifying NNCS with general neural-network controllers called ReachNN based on Bernstein polynomial. More specifically, given an input space and a degree bound, we construct a polynomial approximation for a general neural-network controller based on Bernstein polynomials. For the critical step of estimating the approximation error bound, inspired by the observation that most neural networks are Lipschitz continuous~\cite{ruan2018reachability}, we present two techniques -- a priori theoretical approach based on existing results on Bernstein polynomials and a posteriori approach based on adaptive sampling.
By applying these two techniques together, we are able to capture the behavior of a neural-network controller during verification via Bernstein polynomial with tight error bound estimation. 
Based on the polynomial approximation with the bounded error, we can iteratively compute an overapproximated reachable set of the neural-network controlled system via flowpipes~\cite{Zhao/1992/phd}. By the Stone-Weierstrass theorem~\cite{de1959stone}, our Bernstein polynomial based approach can approximate most neural networks with different activation functions (e.g., ReLU, sigmoid, tanh) to arbitrary precision. Furthermore, as we will illustrate later in Section~\ref{sec:approach}, the approximation error bound can be conveniently calculated. 

\smallskip
\noindent
Our paper makes the following contributions.
\begin{myitemize}
    \item We proposed a Bernstein polynomial based approach to generate high-order approximations for the input-output mapping of general neural-network controllers, which is much tighter than the interval based approaches.
    \item We developed two techniques to analyze the approximation error bound for neural networks with different activation functions and structures based on the Lipschitz continuity of both the network and the approximation. One is based on the theory of Bernstein polynomials and provides a priori insight of the theoretical upper bound of the approximation error, while the other achieves a more accurate estimation in practice via adaptive sampling.
    \item We demonstrated the effectiveness of our approach on multiple benchmarks, showing its capability in handling dynamical systems with various neural-network controllers, including heterogeneous neural networks with multiple types of activation functions. 
    For homogeneous networks, compared with state-of-the-art approaches Sherlock and Verisig, our ReachNN approach can achieve comparable or even better approximation performance, albeit with longer computation time.
\end{myitemize}

The rest of the paper is structured as follows. Section~\ref{sec:problem} introduces the system model, the reachability problem, and more details on the most relevant works. Section~\ref{sec:approach} presents our approach, including the construction of polynomial approximation and the estimation of error bound. Section~\ref{sec:experiment} presents the experimental results. Section~\ref{sec:challenges} provides further discussion of our approach and Section~\ref{sec:conclusion} concludes the paper.

%% file: sections/problem.tex
\section{Problem Statement} \label{sec:problem}


In this section, we describe the reachability of NNCS and a solution framework that computes overapproximations for reachable sets. In the paper, a set of ordered variables $x_1,x_2,\dots,x_n$ is collectively denoted by $x$. For a vector $x$, we denote its $i$-th component by $x_i$.


A NNCS is illustrated in the Figure~\ref{fig:snns_structure}. The plant is the formal model of a physical system or process, defined by an ODE in the form of $\dot{x}=f(x,u)$ such that $x$ are the $n$ state variables and $u$ are the $m$ control inputs. We require that the function $f:\,\reals^m\times \reals^n \rightarrow \reals^m$ is Lipschitz continuous in $x$ and continuous in $u$, in order to guarantee the existence of a unique solution of the ODE from a single initial state (see~\cite{Meiss/2007/Differential}).

The controller in our system is implemented as a feed-forward neural network, which can be defined as a function $\kappa$ that maps the values of $x$ to the control inputs $u$. It consists of $S$ layers, where the first $S-1$ layers are referred as ``hidden layers'' and the $S$-th layer represents the network's output. Specifically, we have
\begin{displaymath}
\kappa(x) = \kappa_S(\kappa_{S-1}(\dots \kappa_1(x; W_1, b_1); W_2, b2); W_L, b_L)
\end{displaymath}
where $W_s$ and $b_s$ for $s=1,2,\dots,S$ are learnable parameters as linear transformations connecting two consecutive layers, which is then followed by an element-wise nonlinear activation function. $\kappa_{i}(z_{s-1}; W_{s-1}, b_{s-1})$ is the function mapping from the output of layer $s-1$ to the output layer $s$ such that $z_{s-1}$ is the output of layer $s-1$.
An illustration of a neural network is given in Fig~\ref{fig:snns_structure}.

A NNCS works in the following way. Given a control time stepsize $\delta_c > 0$, at the time $t = i \delta_c$ for $i=0,1,2,\dots$, the neural network takes the current state $x(i \delta_c)$ as input, computes the input values $u(i \delta_c)$ for the next time step and feeds it back to the plant. More precisely, the plant ODE becomes $\dot{x} = f(x, u(i\delta_c))$ in the time period of $[i\delta_c, (i+1)\delta_c]$ for $i=0,1,2,\dots$. Notice that the controller does not change the state of the system but the dynamics. The formal definition of a NNCS is given as below.

\begin{figure}[tbp]
	\centering	
	\includegraphics[width=0.4\textwidth]{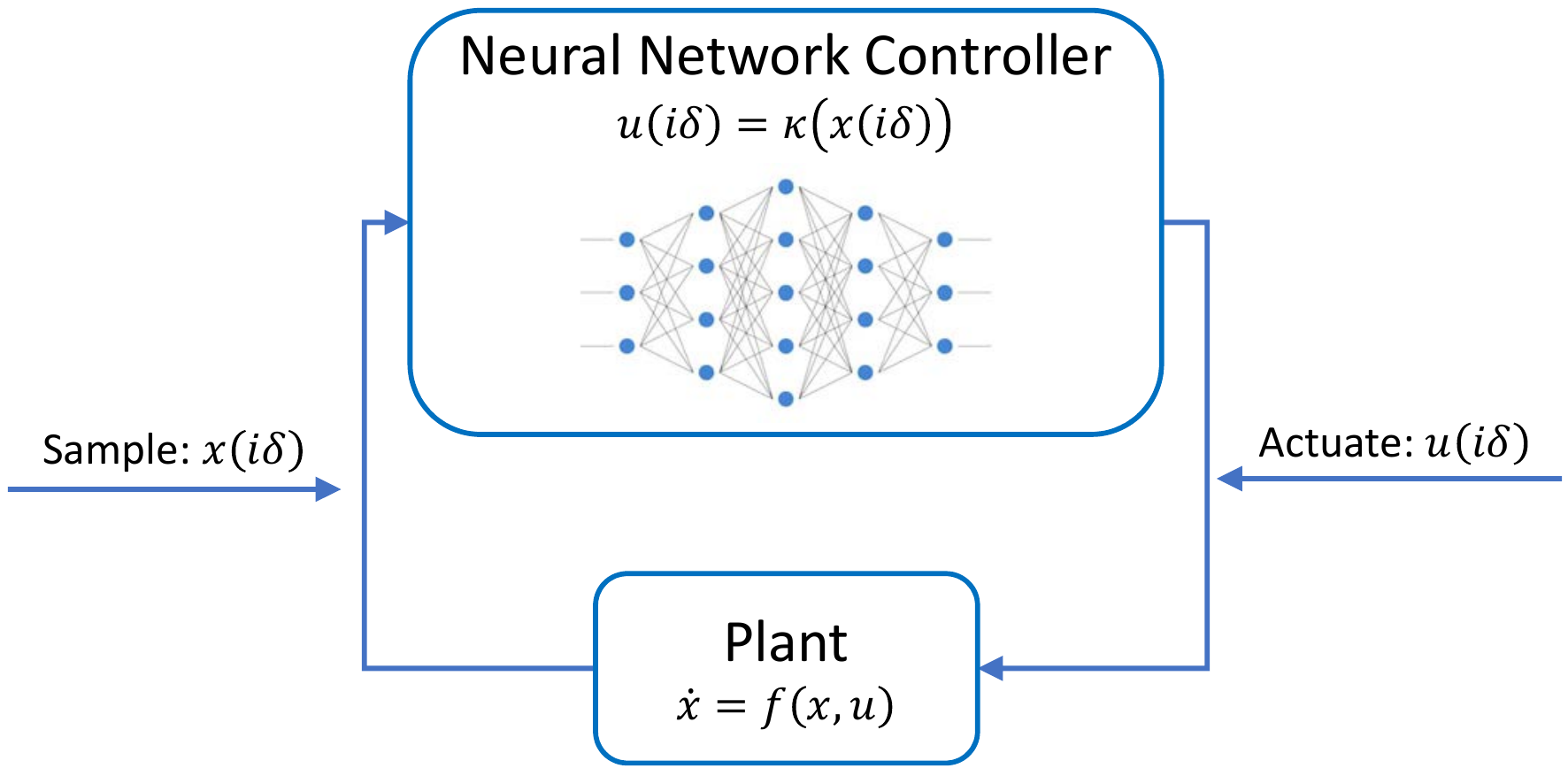}%
	\caption{Neural-network controlled system (NNCS).}
	\label{fig:snns_structure}
\end{figure}

\begin{definition}[Neural-Network Controlled System]
A \emph{neural-network controlled system (NNCS)} can be denoted by a tuple $(\mathcal{X}, \mathcal{U}, F, \kappa, \delta_c, X_0)$, where $\mathcal{X}$ denotes the state space whose dimension is the number of state variables, $\mathcal{U}$ denotes the control input set whose dimension is the number of control inputs, $F$ defines the continuous dynamics $\dot{x} = f(x, u)$, $\kappa: \mathcal{X} \rightarrow \mathcal{U}$ defines the input/output mapping of the neural-network controller, $\delta_c$ is the control stepsize, and $X_0\subseteq \mathcal{X}$ denotes the initial state set.
\end{definition}

Notice that a NNCS is deterministic when the continuous dynamics function $f$ is Lipschitz continuous. The behavior of a NNCS can be defined by its flowmap. The \emph{flowmap} of a system $(\mathcal{X}, \mathcal{U}, F$, $\kappa, \delta_c, X_0)$ is a function $\varphi:X_0 \times \reals_{\geq 0} \rightarrow \mathcal{X}$ that maps an initial state $x_0$ to the state $\varphi(x_0, t)$, which is the system state at the time $t$ from the initial state $x_0$. Given an initial state $x_0$, the flowmap has the following properties for all $i=0,1,2,\dots$: (a) $\varphi$ is the solution of the ODE $\dot{x} = f(x, u(i\delta_c))$ with the initial condition $x(0) = \varphi(x_0, i\delta_c)$ in the time interval $t \in [t - i\delta_c, t - i\delta_c + \delta_c]$; (b) $u(i\delta_c) = \kappa(\varphi(x_0, i\delta_c))$.

We call a state $x$ \emph{reachable} at time $t\geq 0$ on a system $(\mathcal{X}, \mathcal{U}, F$, $\kappa, \delta_c, X_0)$, if and only if there is some $x_0\in X_0$ such that $x = \varphi(x_0, t)$. Then, the set of all reachable states is called the \emph{reachable set} of the system.

\begin{definition}[Reachability Problem]
 The \emph{reachability problem} on a NNCS is to decide whether a given state is reachable or not at time $t\geq 0$.
\end{definition}


In the paper, we focus on the problem of computing the reachable set for a NNCS. Since NNCSs are at least as expressive as nonlinear continuous systems, the reachability problem on NNCSs is undecidable. Although there are numerous existing techniques for analyzing the reachability of linear and nonlinear hybrid systems~\cite{Frehse+/2011/SpaceEx,Chen+/2012/taylor_models,Kong+/2015/dReach,Duggirala+/2015/C2E2,Althoff/2015/CORA}, none of them can be directly applied to NNCS, since equivalent transformation from NNCS to a hybrid automaton is usually very costly due to the large number of locations in the resulting automaton. Even an on-the-fly transformation may lead to a large hybrid automaton in general. Hence, we compute flowpipe overapproximations (or flowpipes) for the reachable sets of NNCS.

Similar to the flowpipe construction techniques for the reachability analysis of hybrid systems, we also seek to compute overapproximations for the reachable segments of NNCS. A continuous dynamics can be handled by the existing tools such as SpaceEx~\cite{Frehse+/2011/SpaceEx} when it is linear, and Flow*~\cite{Chen+/2013/flowstar} or CORA~\cite{Althoff/2015/CORA} when it is nonlinear. The challenge here is to compute an accurate overapproximation for input/output mapping of the neural-network controller in each control step, and we will do it in the following way.
 
Given a bounded input interval $X_I$, we compute a model $(g(x), \epsilon)$ where $\epsilon \geq 0$ such that for any $x\in X_I$, the control input $\kappa(x)$ belongs the set
$\{g(x) + z \,|\, z\in B_{\epsilon}\}$. $g$ is a function of $x$, and $B_{\epsilon}$ denotes the box $[-\epsilon, \epsilon]$ in each dimension.
We provide a summary of the existing works which are close to ours.


\smallskip
\noindent
\textbf{Interval overapproximation.}
The methods described in~\cite{ruan2018reachability,xiang2018reachability} compute intervals as neural-network input/output relation and directly feed these intervals in the reachability analysis. Although they can be applied to more general neural-network controllers,
using interval overapproximation in the reachability analysis cannot capture the dependencies of state variables for each control step, and it is reported in \cite{duttoutputa2018}.

\smallskip
\noindent
\textbf{Exact neural network model.} The approach presented in~\cite{ivanov2018verisig} equivalently transforms the neural-network controller to a hybrid system, then the whole NNCS becomes a hybrid system and the existing analysis methods can be applied. The main limitations of the approach are: (a) the transformation could generate a model whose size is prohibitively large, and (b) it only works on neural networks with sigmoid and tanh activation functions. 

\smallskip
\noindent
\textbf{Polynomial approximation with error bound.}
In~\cite{Dutta_Others__2019__Reachability}, the authors describe a method to produce higher-order sets for neural-network outputs. It is the closest work to ours. In their paper, the approximation model is a piecewise polynomial over the state variables, and the error bound can be well limited when the degrees or pieces of the polynomials are sufficiently high. The main limitation of the method is that it only applies to the neural networks with ReLU activation functions.


%% file: sections/approach.tex
\section{Our Approach} \label{sec:approach}

Our approach exploits high-order set approximation of the output of general types neural network in NNCS reachability analysis via Bernstein polynomials and approximation error bound estimation. The reachable set of NNCS is overapproximated by a finite set of Taylor model flowpipes in our method. The main framework of flowpipe construction is presented in Algorithm~\ref{alg:framework}. Given a NNCS and a bounded time horizon $[0,N\delta_c]$, the algorithm computes $Nk$ flowpipes, each of which is an overapproximation of the reachable set in a small time interval, and the union of the flowpipes is an overapproximation of the reachable set in the time interval of $[0,N\delta_c]$. As stated in Theorem~\ref{thm:overapproximation}, this fact is guaranteed by (a) the flowpipes are reachable set overapproximations for the continuous dynamics in each iteration, and (b) the set $U_i$ is an overapproximation of the neural-network controller output. Notice that this framework is also used in~\cite{Dutta_Others__2019__Reachability}. However, we present a technique to compute $U_i$ for more general neural networks.


\begin{algorithm} [tbp]
    \SetAlgoLined
    \KwData{NNCS $(\mathcal{X}, \mathcal{U}, F, \kappa, \delta_c, X_0)$, time horizon $N\delta_c$}
    \KwResult{Flowpipes}
     Flowpipes $\gets \emptyset$\;
    \For{$i \gets 0$ to $N-1$}{
        Compute a set $U_i$ which contains the value of $\kappa(x)$ for all $x \in X_i$\;
        Compute the flowpipes $\flowpipe_1,\dots,\flowpipe_k$ for the continuous dynamics $\dot{x} = f(x, y)$ with $x(0)\in X_i$ and $y\in U_i$\;
        Evaluate the flowpipe $X_{i+1}$ based on $\flowpipe_k$ for the reachable set at $t = (i+1)\delta_c$\;
        Flowpipes $\gets$ Flowpipes $\cup\, \{\flowpipe_1,\dots,\flowpipe_k\}$\;
    }
    \Return Flowpipes\;
\caption{Flowpipe construction for NNCS}
\label{alg:framework}
\end{algorithm}


Let us briefly revisit the technique of computing Taylor model flowpipes for continuous dynamics.

\smallskip
\noindent\textbf{Taylor model.}
Taylor models are introduced to provide higher-order enclosures for the solutions of nonlinear ODEs (see~\cite{Berz+Makino/1998/Verified}), and then extended to overapproximate reachable sets for hybrid systems~\cite{Chen+/2012/taylor_models} and solve optimization problems~\cite{Makino+Berz/2005/TM_Optimization}. A \emph{Taylor Model (TM)} is denoted by a pair $(p, I)$ where $p$ is a (vector-valued) polynomial over a set of variables $x$ and $I$ is an (vector-valued) interval. A continuous function $f(x)$ can be overapproximated by a TM $(p(x), I)$ over a domain $D$ in the way that $f(x) \in p(x) + I$ for all $x\in D$. When the approximation $p$ is close to $f$, $I$ can be made very small.


\smallskip
\noindent\textbf{TM flowpipe construction.}
The technique of TM flowpipe construction is introduced to compute overapproximations for the reachable set segments of continuous dynamics. Given an ODE $\dot{x} = f(x)$, an initial set $X_0$ and a time horizon $[0,T]$, the method computes a finite set of TMs $\flowpipe_1,\dots,\flowpipe_k$ such that for each $i=1,\dots,k$, the flowpipe $\flowpipe_i$ contains the reachable set in the time interval of $[t_i,t_{i+1}]$, and $\bigcup_{i=1}^k [t_i,t_{i+1}] = [0,T]$, i.e., the union of the flowpipes is an overapproximation of the reachable set in the given time horizon. The standard TM flowpipe construction is described in~\cite{Berz+Makino/1998/Verified}, and it is adapted in~\cite{Chen/2015/phd,Chen+Sankaranarayanan/2016/Decomposed} for better handling the dynamical systems.


\smallskip
\noindent

The main contribution of our paper is a novel approach to compute a higher-order overapproximation $U_i$ for the output set of neural networks with a more general form of activation functions, including such as ReLU, sigmoid, and tanh. We show that this approach can generate accurate reachable set overapproximations for NNCSs, in combination with the TM flowpipe construction framework. Our main idea can be described as follows.

Given a neural-network controller with a single output, we assume that its input/output mapping is defined by a function $\kappa$ and its input interval is defined by a set $X_i$. In the $i$-th (control) step, we seek to compute a TM $U_i = P(x) + [-\bar{\varepsilon},\bar{\varepsilon}]$ such that
\begin{equation}\label{eq:output_overapproximation}
    \kappa(x) \in P(x) + [-\bar{\varepsilon},\bar{\varepsilon}] \mbox{\quad for all\quad } x\in X_i.
\end{equation}
 Hence, the TM is an overapproximation of the neural network output. In the paper, we compute $P$ as a Bernstein polynomial with bounded degrees. Since $\kappa$ is a continuous function that can be approximated by a Bernstein polynomial to arbitrary precision, according to the Stone-Weierstrass theorem~\cite{de1959stone}, we can always ensure that such $P$ exists.


In our reachability computation, the set $X_i$ is given by a TM flowpipe. To obtain a TM for the neural network output set, we use TM arithmetic to evaluate a TM for $P(x) + [-\bar{\varepsilon},\bar{\varepsilon}]$ with $x\in X_i$. Then, the polynomial part of the resulting TM can be viewed as an approximation of the mapping from the initial set to the neural network output, and the remainder contains the error. Such representation can partially keep the variable dependencies and much better limit the overestimation accumulation than the methods that purely use interval arithmetic.

\begin{theorem}\label{thm:overapproximation}
   The union of the flowpipes computed by Algorithm~\ref{alg:framework} is an overapproximation of the reachable set of the system in the time horizon of $[0, N\delta_c]$, if the flowpipes are overapproximations of the ODE solutions and the TM $U_i$ in every step satisfies (\ref{eq:output_overapproximation}).
\end{theorem}

\begin{remark}
    In our framework, Taylor models can also be considered as a candidate of high-order approximation for a neural network's input-output mapping. However, comparing with Bernstein polynomial based approach adopted in this paper, Taylor models suffer from two main limitations: (1) The validity of Taylor models relies on the function differentiability, while ReLU neural networks are not differentiable. Thus Taylor models cannot handle a large number of neural networks; (2) There is no theoretical upper bound estimation for Taylor models, which further limits the rationality of using Taylor models.
\end{remark}

\subsection{Bernstein Polynomials for Approximation}

\begin{definition}[Bernstein Polynomials] \label{de:bernstein}
    Let $d=(d_1,\cdots,d_m) \in \mathbb{N}^m$ and $f$ be a function of $x=(x_1.\cdots,x_m)$ over $I=[0,1]^m$. The polynomials
    \begin{displaymath}
        B_{f,d}(x)\ =\ \sum_{\mathclap{\substack{0\leq k_j\leq d_j \\ j\in\{1,\cdots, m\}}}} f(\frac{k_1}{d_1},\cdots,\frac{k_m}{d_m}) \prod_{j=1}^{m}\left(\binom{d_j}{k_j}x_j^{k_j}(1-x_j)^{d_j-k_j}\right)
    \end{displaymath}
    are called the \emph{Bernstein polynomials} of $f$ under the degree $d$.
\end{definition}

We then construct $P(x)$ over $X$ by a series of linear transformation based on Bernstein polynomials. Assume that $X=[l_1,u_1]\times\cdots \times [l_m,u_m]$. Let $x'=(x'_1,\cdots ,x'_m)$, where
\begin{displaymath}
    x'_j=(x_j-l_j)/(u_j-l_j), \quad j=1,\cdots , m
\end{displaymath}
and 
\begin{equation} \label{eq:transformed_kappa}
    \kappa'(x')
    =
    \kappa(x)
    =
    \kappa
    \left(
    \begin{pmatrix} u_1-l_1 & \cdots & 0 \\ \vdots & \ddots & \vdots \\ 0 & \cdots &  u_m-l_m \end{pmatrix}
    x'
    + 
    \begin{pmatrix}
		l_1  \\
		\vdots \\
		l_m
    \end{pmatrix}
    \right).
\end{equation}
It is easy to see that $\kappa'$ is defined over $I$. For $d=(d_1,\cdots, d_m) \in \mathbb{N}^m$, let $B_{\kappa',d}(x')$ be the Bernstein polynomials of $\kappa'(x')$.
We construct the polynomial approximation for $\kappa$ as: 
\begin{equation} \label{eq:poly_approx_construction}
    P_{\kappa,d}(x)
    = 
    B_{\kappa',d}
    \left(
        \begin{pmatrix} \frac{1}{u_1-l_1}  & \cdots & 0 \\ \vdots & \ddots & \vdots \\ 0 & \cdots &  \frac{1}{u_m-l_m} \end{pmatrix}
    x
    - 
    \begin{pmatrix}
		\frac{l_1}{u_1-l_1} \\
		\vdots \\
		\frac{l_m}{u_m-l_m} 
    \end{pmatrix}
    \right).
\end{equation}

When we want to compute a Bernstein polynomial over a non-interval domain $I$, we may consider an interval enclosure of $I$, since we only need to ensure that the polynomial is valid on the domain and it is sufficient to take its superset. Hence, in Algorithm~\ref{alg:framework}, the Bernstein polynomial(s) in $U_i$ are computed based on an interval enclosure of $X_i$.



%% file: sections/error.tex
\subsection{Approximation Error Estimation}

After we obtain the approximation of the neural network controller, a certain question is how to estimate a valid bound for approximation error $\varepsilon$ such that Theorem \ref{thm:overapproximation} holds. Namely, from any given initial state set $X$, the reachable set of the perturbed system $\dot{x} = f(x,P_{\kappa,d},\varepsilon)$, $\varepsilon \in [-\bar{\varepsilon},\bar{\varepsilon}]$ at any time $t\in [0,\delta_c]$ is a superset of the one of the NNCS with ODE $f(x,\kappa)$. A sufficient condition can be derived based on the theory of differential inclusive \cite{smirnov2002introduction}:
\begin{lemma}
    Given any state set $X$, let $P_{\kappa,d}$ be the polynomial approximation of $\kappa$ with respect to the degree $d$ defined as Equation \eqref{eq:poly_approx_construction}. For any time $t\in [0,\delta_c]$, the reachable set of the perturbed system $\dot{x} = f(x,P_{\kappa,d}+\varepsilon)$, $\varepsilon \in [-\bar{\varepsilon},\bar{\varepsilon}]$ is a superset of the one of the NNCS with ODE $f(x,\kappa)$ from $X$, if 
    \begin{equation}\label{eq:condition}
        \kappa(x) \ \in \{\ u\ |\ u=P_{\kappa,d}(x) + \varepsilon ,\ \ \varepsilon\in [-\bar{\varepsilon},\bar{\varepsilon}]\}, \quad \forall x\in X.
    \end{equation}
\end{lemma}

Intuitively, the approximation error interval $E=[-\bar{\varepsilon},\bar{\varepsilon}]$ has a significant impact on the reachable set overapproximation, namely a tighter $\bar{\varepsilon}$ can lead to a more accurate reachable set estimation. In this section, we will introduce two approaches to estimate $\bar{\varepsilon}$, namely theoretical error estimation and sampling-based error estimation. The former gives us a priori insight of how precise the approximation is, while the latter one helps us to obtain a much tighter error estimation in practice.

\smallskip
\noindent
\textbf{Compute a Lipschitz constant.} 
We start from computing the Lipschitz constant of a neural network, since Lipschize constant plays a key role in both of our two approaches, which we will see later. 

\begin{definition}
    A real-valued function $f: X \rightarrow \mathbb{R}$ is called Lipschitz continuous over $X\subseteq \mathbb{R}^m$, if there exists a non-negative real $L$, such that for any $x,x' \in X$:
    \begin{displaymath}
        \left\|f(x)-f(x')\right\| \leq L\left\|x-x'\right\|.
    \end{displaymath}
    Any such $L$ is called a Lipschitz constant of $f$ over $X$.
\end{definition}

Recent work has shown that a large number of neural networks are Lipschitz continuous, such as the fully-connected neural networks with ReLU, sigmoid, and tanh activation functions and the estimation of Lipschitz constant upper bound for a neural network has been preliminary discussed in \cite{ruan2018reachability,szegedy2013intriguing}.
\begin{lemma}[Lipschitz constant for sigmoid/tanh/ReLU~\cite{ruan2018reachability,szegedy2013intriguing}] \label{lma:basic_lipschitz}
    Convolutional or fully connected layers with the sigmoid activation function $\mathcal{S}(Wx+b)$, hyperbolic tangent (tanh) activation function $\mathcal{T}(Wx+b)$, and ReLU activation function $\mathcal{R}(Wx+b)$ have $\frac{1}{4}\|W\|$, $\|W\|$, $\|W\|$ as their Lipschitz constants, respectively.
\end{lemma}


Based on Lemma \ref{lma:basic_lipschitz}, we further improve the Lipschitz constant upper bound estimation.
Specifically, we consider a layer of neural network with $n$ neurons shown in Figure \ref{fig:range}, where $W$ and $b$ denote the weight and the bias that are applied on the output of the previous layer. \emph{Input Interval} and \emph{Output Interval} denote the variable space before and after applied by the activation functions of this layer, respectively. Assume $X=[l_1,u_1]\times \cdots \times [l_n,u_n]$ be the Input Interval. 

We first discuss layers with sigmoid/tanh activation functions based on the following conclusion:
\begin{lemma} \label{lma:diff_lipschitz} \cite{royden1968real}
    Given a function $f: X \rightarrow \mathbb{R}^m$, if $\left\|\partial f/\partial x\right\|\leq L$ over $X$, then $f$ is Lipschitz continuous and $L$ is a Lipschitz constant.
\end{lemma}
\noindent
\textbf{Sigmoid.} For a layer with sigmoid activation function $\mathcal{S}(y)=1/(1+e^{-y})$ with $y=Wx+b$ and $y\in X$, we have
\begin{equation} \label{eq:sigm_lips}
    \begin{aligned}
        & \left\|\frac{\partial \mathcal{S}(x)}{\partial x}\right\| 
        =  \left\|\frac{\partial \mathcal{S}(y)}{\partial y} \frac{\partial y}{\partial x}\right\| \leq  \left\|\frac{\partial \mathcal{S}(y)}{\partial y}\right\| \left\|\frac{\partial y}{\partial x}\right\|\\
        = &  
        \left\|\diag(\mathcal{S}(y_1)(1-\mathcal{S}(y_1)),\cdots, \mathcal{S}(y_n)(1-\mathcal{S}(y_n)))\right\| 
        \left\|W\right\|\\
       \leq &  \max_{1\leq i \leq n}\sup_{a_i\leq \mathcal{S}(y_i) \leq b_i}  \{\mathcal{S}(y_i)(1-\mathcal{S}(y_i))\} \left\|W\right\|\\
       = &  \max_{1\leq i \leq n} \left(\frac{1}{4}-\left(\frac{\sgn(\mathcal{S}(a_i)-0.5)+\sgn(\mathcal{S}(b_i)-0.5)}{2}\right)^2 \right. \cdot\\
       & \qquad \quad \left. \min\left\{\left(\frac{1}{2}-\mathcal{S}(a_i)\right)^2,\left(\frac{1}{2}-\mathcal{S}(b_i)\right)^2\right\}\right) \left\|W\right\|
    \end{aligned}
\end{equation}

\noindent
\textbf{Hyperbolic tangent.} For a layer with hyperbolic tangent activation function $\mathcal{T}(y)=2/(1+e^{-2y})-1$ with $y=Wx+b$ and $y\in X$, we have
\begin{equation} \label{eq:tanh_lips}
    \begin{aligned}
        & \left\|\frac{\partial \mathcal{T}(x)}{\partial x}\right\| 
       =   \left\|\frac{\partial \mathcal{T}(y)}{\partial y} \frac{\partial y}{\partial x}\right\| \leq  \left\|\frac{\partial \mathcal{T}(y)}{\partial y}\right\| \left\|\frac{\partial y}{\partial x}\right\|\\
       =  & 
            \left\|\diag(1-(\mathcal{T}(y_1))^2,\cdots, 1-(\mathcal{T}(y_n))^2)\right\| 
            \left\|W\right\|\\
        =  &\max_{1\leq i \leq n} \sup_{a_i\leq \mathcal{T}(y_i) \leq b_i} \{1-(\mathcal{T}(y_i))^2\} \left\|W\right\|\\
       = &  
        \max_{1\leq i \leq n} \left(1-\left(\frac{\sgn(\mathcal{T}(a_i))+\sgn(\mathcal{T}(b_i))}{2}\right)^2 \right. \cdot\\
        & \qquad \quad \left.
        \min\{(\mathcal{T}(a_i))^2,(\mathcal{T}(b_i))^2\}\right) \left\|W\right\| 
    \end{aligned}
\end{equation}

For ReLU networks, we try to derive a Lipschitz constant directly based on its definition:

\noindent
\textbf{ReLU.} For a layer with ReLU activation function $\mathcal{R}(y)=\max\{0,y\}$ with $y=Wx+b$ and $y\in X$, we have
\begin{equation} \label{eq:relu_lips}
    \begin{aligned}
        &\sup_{x_1\neq x_2}\frac{\left\|\mathcal{R}(x_1)-\mathcal{R}(x_2)\right\|}{\left\|x_1-x_2\right\|} \\
        = & \sup_{x_1\neq x_2}\frac{\left\|\begin{pmatrix}\max\{0,W_1 x_1+u_1\} - \max\{0,W_1 x_2+u_1\}\\ \cdots \\ \max\{0,W_n x_1+u_n\} - \max\{0,W_n x_2+u_n\}\end{pmatrix}\right\|}{\left\|x_1-x_2\right\|} \\
        \leq & \sup_{x_1\neq x_2}\frac{\left\|\begin{pmatrix} \left(\frac{1+\sgn(u_1)}{2}\right)^2W_1|x_1-x_2| \\ \cdots \\ \left(\frac{1+\sgn(u_n)}{2}\right)^2W_n|x_1-x_2| \end{pmatrix}\right\|}{\left\|x_1-x_2\right\|} \\
        \leq & \left\|\left(\left(\frac{1+\sgn(u_1)}{2}\right)^2W_1, \cdots , \left(\frac{1+\sgn(u_n)}{2}\right)^2W_n \right)\right\|
    \end{aligned}.
\end{equation}

\begin{figure}[tbp]
	\centering	
	\includegraphics[width=0.3\textwidth]{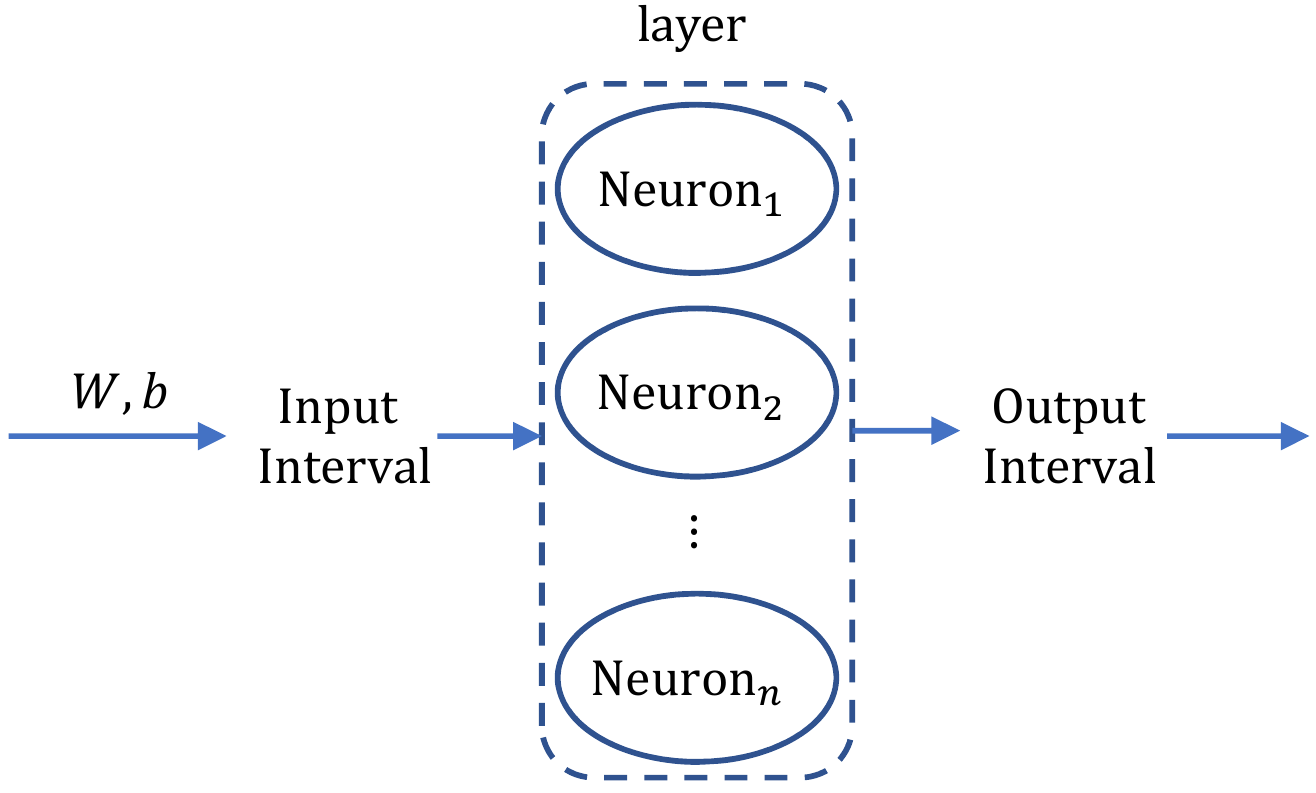}%
	\caption{Schematic diagram of Input Interval and Output Interval of a layer.}
	\label{fig:range}
\end{figure}

\begin{algorithm} [tbp]
    \SetAlgoLined
    \KwData{Neural network $\kappa$, input space $X$}
    \KwResult{Lipschitz constant $L$}
    $S \gets$ the number of layers of $\kappa$\;
    $L \gets 1$\;
    $\text{OutputInterval}_0 \gets X$\;
    \For{$s \gets 1$ to $S$}{
        $act \gets$ the activation function of the $s$-th layer\;
        $\text{InputInterval}_i \gets $ ComputeInputInterval($\text{OutputInterval}_{i-1}$)\;
        $L_i \gets $ ComputeLipchitzConstant($act$, $W$, $\text{InputInterval}_i$)\;
        $L \gets L\cdot L_i$\;
        $\text{OutputInterval}_i \gets $ ComputeOutputInterval($\text{InputInterval}_{i}$)\;
    }
    \Return{$L$}\;
\caption{Compute a Lipschitz constant for $\kappa$}
\label{alg:lips}
\end{algorithm}

In Algorithm~\ref{alg:lips}, we first do the initialization (line 1-3) by letting the variable denoting Lipschitz constant $L=1$. Note that $\text{OutputInterval}_i (i=1,\cdots ,S$) denotes the input interval and output interval of layer $i$, as shown in Figure \ref{fig:range}. For convenience, we let $\text{OutputInterval}_0$ be the input space $X$. Then we do the layer-by-layer interval analysis and compute the corresponding Lipschitz constant (line 4-10). For each layer $s$, the function \textbf{ComputeInputInterval} is first invoked to compute the $\text{InputInterval}_s$ (line 6). Then the Lipschitz constant $L_s$ of layer $s$ is evaluated by the function \textbf{ComputeLipchitzConstant}, namely Equation \eqref{eq:sigm_lips}, \eqref{eq:tanh_lips}, \eqref{eq:relu_lips} in terms of the  activation function $act$ of this layer (line 7). $L$ is updated to the current layer by multiplying $L_s$ (line 8). Finally, $\text{OutputInterval}_s$ is computed by the function \textbf{ComputeOutputInterval} (line 9). Note that the implementation of \textbf{ComputeInputInterval} and \textbf{ComputeOutputInterval} is a typical interval analysis problem of neural networks, which have been adequately studied in \cite{duttoutputa2018,xiang2018reachability,ruan2018reachability}. We do not go into details here due to the space limit.

\noindent
\textbf{Naive theoretical error (T-error) estimation.}
After obtaining a Lipschize constant of $\kappa$, we can directly leverage the existing result on Bernstein polynomials for Lipschitz continuous functions to derive the error of our polynomial approximation.

\begin{lemma} \cite{lorentz2013bernstein}
     Assume $f$ is a Lipschitz continuous function of $x=(x_1.\cdots,x_m)$ over $I=[0,1]^m$ with a Lipschitz constant $L$. Let $d=(d_1,\cdots,d_m) \in \mathbb{N}^m$ and $B_{f,d}$ be the Bernstein polynomials of $f$ under the degree $d$. Then we have
    \begin{equation} \label{eq:error_lips}
        \left\|B_{f,d}(x)-f(x)\right\|\leq \frac{L}{2}\left( \sum_{j=1}^{m}({1}/{d_j}) \right)^{\frac{1}{2}}, \quad \forall x\in I.
    \end{equation}
\end{lemma}

\begin{theorem}[T-Error Estimation] \label{thm:t_error_bound}
    Assume $\kappa$ is a Lipschitz continuous function of $x=(x_1,\cdots,x_m)$ over $X=[l_1,u_1]\times\cdots \times [l_m,u_m]$ with a Lipschitz constant $L$. Let $P_{\kappa,d}$ be the polynomial approximation of $\kappa$ that is defined as Equation \eqref{eq:poly_approx_construction} with the degree $d=(d_1,\cdots, d_m) \in \mathbb{N}^m$. 
    Let
    \begin{equation} \label{eq:t_error}
        \begin{aligned}
        \bar{\varepsilon}_t = \frac{L}{2}\left( \sum_{j=1}^{m}\frac{1}{d_j}\right)^{\frac{1}{2}}\max_{j\in\{1,\cdots, m\}}\{u_j-l_j\}.
        \end{aligned}
    \end{equation}
    then $\bar{\varepsilon}_t$ satisfies \eqref{eq:condition}, namely,
    \begin{displaymath}
        \kappa(x) \ \in \{\ u\ |\ u=P_{\kappa,d}(x) + \varepsilon ,\ \ \varepsilon\in [-\bar{\varepsilon}_t,\bar{\varepsilon}_t]\}, \quad \forall x\in X.
    \end{displaymath}
\end{theorem}

\begin{proof}
    First, by Equation \eqref{eq:transformed_kappa} we know that 
    \begin{displaymath}
        \left\|{\partial x}/{\partial x'}\right\| = \max_{j\in\{1,\cdots, m\}}\{u_j-l_j\}
    \end{displaymath}
    is a Lipschitz constant $L_{x(x')}$ of the function $x(x')$. Note that $\kappa'=\kappa \circ x$, then we can obtain the Lipschitz constant of $\kappa'(x')$ by
    \begin{displaymath}
        \begin{aligned}
        L_{\kappa'(x')} = L_{\kappa(x)}\cdot L_{x(x')}=  L\max_{j\in\{1,\cdots, m\}}\{u_j-l_j\}.
        \end{aligned}
    \end{displaymath}
    By Equation \eqref{eq:error_lips}, $\forall x\in X$, we have:
    \begin{align*}
            & \left\|P_{\kappa,d}(x){-}\kappa(x)\right\| =
            \left\|B_{\kappa',d}(x') {-} \kappa'(x')\right\|
            \leq  \frac{L_{\kappa'(x')}}{2}\left( \sum_{j=1}^{m}\frac{1}{d_j}\right)^{\frac{1}{2}}. \qedhere
    \end{align*}
\end{proof}

\smallskip
\noindent

\textbf{Adaptive sampling-based error (S-error) estimation.} 
While Theorem \ref{thm:t_error_bound} can help derive an approximation error bound easily, such bounds are often over-conservative in practice. Thus we propose an alternative sampling-based approach to estimate the approximation error. 
For a given box $X=[l_1,u_1]\times\cdots \times [l_m,u_m]$, we perform a grid-based partition based on an integer vector $p=(p_1,\cdots ,p_m)$. That is, we partition $X$ into a set of boxes $X=\bigcup_{0\leq k \leq p-1}B_k$, where $0\leq k \leq p-1$ is the abbreviation for $k=(k_1,\cdots ,k_m), 0\leq k_j \leq p_j-1,\ 1\leq j\leq m$,
and for any $k$,
\begin{displaymath}
    \begin{aligned}
    B_k = &[l_1+\frac{k_1}{p_1}(u_1-l_1),l_1+\frac{k_1+1}{p_1}(u_1-l_1)]\times  \cdots \times \\
    &[l_m+\frac{k_m}{p_m}(u_m-l_m),l_j+\frac{k_m+1}{p_m}(u_m-l_m)].
    \end{aligned}
\end{displaymath}

It is easy to see that the largest error bound of all the boxes is a valid error bound over $X$.
\begin{lemma} \label{lma:partition}
    Assume $\kappa$ is a continuous function of $x=(x_1,\cdots,x_m)$ over $X=[l_1,u_1]\times\cdots \times [l_m,u_m]$. Let $P_{\kappa,d}$ be the polynomial approximation of $\kappa$ that is defined as Equation \eqref{eq:poly_approx_construction} with the degree $d=(d_1,\cdots, d_m) \in \mathbb{N}^m$. Let $\{B_k\}$ be the box partition of $X$ in terms of $p$, and $\bar{\varepsilon}_k$ be the approximation error bound of $P_{\kappa,d}$ over the box $B_k$. We then have $\forall x\in X$,
    \begin{displaymath}
        \kappa(x) \ \in \{\ u\ |\ u=P_{\kappa,d}(x) + \varepsilon ,\ \varepsilon\in [-\max_{0\leq k \leq p-1}\bar{\varepsilon}_k,\max_{0\leq k \leq p-1}\bar{\varepsilon}_k]\}.
    \end{displaymath}
\end{lemma}

Leveraging the Lipschitz continuity of $\kappa$, we can estimate the local error bound $\bar{\varepsilon}_k$ for box $B_k$ by sampling the value of $\kappa$ and $P_{\kappa ,d}$ at the box center.
\begin{lemma} \label{lma:partition_error}
    Assume $\kappa$ is a Lipschitz continuous function of $x=(x_1,\cdots,x_m)$ over box $B_k$ with a Lipschitz constant $L$. Let $P_{\kappa,d}$ be the polynomial approximation of $\kappa$ that is defined as Equation \eqref{eq:poly_approx_construction} with the degree $d=(d_1,\cdots, d_m) \in \mathbb{N}^m$. Let 
    \begin{equation} \label{eq:sample_point}
        c_k = \left(l_1+\frac{2k_1+1}{2p_1}(u_1-l_1),\cdots , l_m+\frac{2k_m+1}{2p_m}(u_m-l_m)\right)
    \end{equation}
    be the center of $B_k$. 
    For $x\in B_k$, we have
    \begin{equation} \label{eq:partition_error}
        \left\|P_{\kappa,d}(x){-}\kappa(x)\right\| \leq L\sqrt{\sum_{j=1}^m(\frac{u_j{-}l_j}{p_j})^2} + \left\|P_{\kappa,d}(c_k){-}\kappa(c_k)\right\|.
    \end{equation}
\end{lemma}

\begin{proof}
    By \cite{brown1987lipschitz}, we know the Bernstein-based approximation $P_{\kappa,d}$ is also Lipschitz continuous with the Lipschitz constant $L$. Then for $x\in B_k$, we have
    \begin{align*}
            & \quad \left\|P_{\kappa,d}(x)-\kappa(x)\right\| \\
            \leq & \quad \left\|P_{\kappa,d}(x)-P_{\kappa,d}(c_k)\right\| + \left\|P_{\kappa,d}(c_k)-\kappa(c_k)\right\| + \left\|\kappa(c_k)-\kappa(x)\right\| \\
            \leq & \quad L\max_{x\in B_k}\left\|x-c\right\| + \left\|P_{\kappa,d}(c_k)-\kappa(c_k)\right\| +
            L\max_{x\in B_k}\left\|x-c_k\right\| \\
            = & \quad L\sqrt{\sum_{j=1}^m(\frac{u_j-l_j}{p_j})^2} + \left\|P_{\kappa,d}(c_k)-\kappa(c_k)\right\|.\qedhere
    \end{align*}
\end{proof}

Combining Lemma \ref{lma:partition} and Lemma \ref{lma:partition_error}, we can derive the sampling-based error bound over $X$.
\begin{theorem}[S-Error Estimation] \label{thm:s_error_bound}
    Assume $\kappa$ is a Lipschitz continuous function of $x=(x_1,\cdots,x_m)$ over $X=[l_1,u_1]\times\cdots \times [l_m,u_m]$ with a Lipschitz constant $L$. Let $P_{\kappa,d}$ be the polynomial approximation of $\kappa$ that is defined as Equation \eqref{eq:poly_approx_construction} with the degree $d=(d_1,\cdots, d_m) \in \mathbb{N}^m$, and $X_c = \{c_k\}_{0\leq k \leq p-1}$ be the sampling set with any given positive integer vector $p=(p_1,\cdots ,p_m)$.
    Let 
    \begin{equation} \label{eq:t_error}
        \begin{aligned}
        \bar{\varepsilon}_s(p) =  L\sqrt{\sum_{j=1}^m(\frac{u_j{-}l_j}{p_j})^2} + \max_{0\leq k \leq p-1} \left\|P_{\kappa,d}(c_k){-}\kappa(c_k)\right\|.
        \end{aligned}
    \end{equation}
    Then, $\bar{\varepsilon}_s$(p) satisfies \eqref{eq:condition}, namely,
    \begin{displaymath}
        \kappa(x) \ \in \{\ u\ |\ u=P_{\kappa,d}(x) + \varepsilon ,\ \ \varepsilon\in [-\bar{\varepsilon}_s(p),\bar{\varepsilon}_s(p)]\}, \quad \forall x\in X.
    \end{displaymath}
\end{theorem}

\begin{theorem} [Convergence of $\bar{\varepsilon}_s$] \label{thm:error_convergence}
    Assume $\kappa$ is a Lipschitz continuous function of $x=(x_1,\cdots,x_m)$ over $X=[l_1,u_1]\times\cdots \times [l_m,u_m]$ with a Lipschitz constant $L$. Let $P_{\kappa,d}$ be the polynomial approximation of $\kappa$ that is defined as Equation \eqref{eq:poly_approx_construction} with the degree $d=(d_1,\cdots, d_m) \in \mathbb{N}^m$. Let $\bar{\varepsilon}_{best}=\max_{x\in X} \left\|P_{\kappa,d}(x)-\kappa\right(x)\|$ be the exact error, we have
    \begin{displaymath}
        \lim_{p\rightarrow \infty}\bar{\varepsilon}_s(p) \rightarrow \bar{\varepsilon}_{best}.
    \end{displaymath}
\end{theorem}
\begin{proof}
    Let $\delta(p)= L\sqrt{\sum_{j=1}^m(\frac{u_j{-}l_j}{p_j})^2}$, we have
    \begin{align*}
        \left|\bar{\varepsilon}_s(p)-\bar{\varepsilon}_{best}\right| \leq
        \left|\bar{\varepsilon}_s(p)-\max_{0\leq k \leq p-1} \left\|P_{\kappa,d}(c_k){-}\kappa(c_k)\right\|\right| = \delta(p)
    \end{align*}
    Since $\delta(p)\rightarrow 0$, when $p\rightarrow \infty$, the theorem holds.
\end{proof}
Note that $\delta(p)$ actually specifies the difference between the exact error and S-error. Thus we call it sampling error precision.

\smallskip
\noindent
\textbf{[Adaptive sampling]} By Theorem \ref{thm:error_convergence}, we can always make $\delta(p)$ arbitrarily small by increasing $p$. Then the error bound is mainly determined by the sample difference over $X_c$. Thus, if our polynomial approximation $P_{\kappa,d}$ regresses $\kappa$ well, we can expect to obtain a tight error bound estimation. To bound the impact of $\delta(p)$, we set a hyper parameter $\bar{\delta}$ as its upper bound. Specifically, given a box $X=[l_1,u_1]\times\cdots \times [l_m,u_m]$, we can adaptively change $p$ to sample fewer times while ensuring $\delta(p)\leq  \bar{\delta}$.
\begin{proposition}
    Given a box $X=[l_1,u_1]\times\cdots \times [l_m,u_m]$ and a positive number $ \bar{\delta}$, let
    \begin{displaymath}
        p_j = \left\lceil{L(u_j-l_j)\sqrt{m}/ \bar{\delta}}\right\rceil, \quad j=1,\cdots ,m.
    \end{displaymath}
    Then we have $\delta(p)\leq \bar{\delta}$.
\end{proposition}
In practice, we can directly use $\bar{\varepsilon}_s$ as $\bar{\varepsilon}$ by specifying a small $\bar{\delta}$, since S-error is more precise. It is worthy noting that a small $\bar{\delta}$ may lead to a large number of sampling points and thus can be time consuming. In our implementation, we mitigate this runtime overhead by computing the sampling errors at $c_k$ for each $B_k$ in parallel.

%% file: sections/experiment.tex
\section{Experiments} \label{sec:experiment}

We implemented a prototype tool and used it in cooperation with Flow*. In our experiments, we first compare our approach ReachNN with the interval overapproximation approach on an example with a heterogeneous neural-network controller that has ReLU and tanh as activation functions. Then we compare our approach with Sherlock \cite{Dutta_Others__2019__Reachability} and Verisig \cite{ivanov2018verisig} on multiple benchmarks collected from the related works. All our experiments were run on a desktop, with 12-core 3.60 GHz Intel Core i7~\footnote{ The experiments are not memory bounded.}.

\subsection{Illustrating Example}

Consider the following nonlinear control system~\cite{Galle19}:
\begin{displaymath}
    \dot{x}_1=x_2,\quad \dot{x}_2=ux_2^2-x_1,
\end{displaymath}
where $u$ is computed from a heterogeneous neural-network controller $\kappa$ that has two hidden layers, twenty neurons in each layer, and ReLU and tanh as activation functions. Given a control stepsize $\delta_c = 0.2$, we hope to verify whether the system will reach $[0,0.2]\times[0.05,0.3]$ from the initial set $[0.8,0.9]\times[0.5,0.6]$. Since the state-of-art verification tools focus on NNCSs with single-type activation function and interval overapproximation is the only presented work that can be extended to heterogeneous neural networks, we compare our method with interval overapproximation in the illustrating example.


Figure~\ref{fig:motivating-tanh-interval} shows the result of the reachability analysis with interval overapproximation, while ReachNN's result is shown in Figure~\ref{fig:motivating-tanh}. Red curves denote the simulation results of the system evolution from $100$ different initial states, which are randomly sampled from the initial set. Green rectangles are the constructed flowpipes as the overapproximation of the reachable state set. The blue rectangle is the goal area. We conduct the reachability analysis from the initial set by combining our neural network approximation approach with Flow*. As shown in Figure~\ref{fig:motivating-tanh-interval}, interval overapproximation based approach yields over-loose reachable set approximation, which makes the flowpipes grow uncontrollably and quickly exceeds the tolerance ($25$ steps). On the contrary, ReachNN can provide a much tighter estimation for all control steps and thus successfully prove the reachability property.

\begin{figure}[tbp]
	\centering	
	\subfloat[][Interval]{%
		\includegraphics[width=0.23\textwidth]{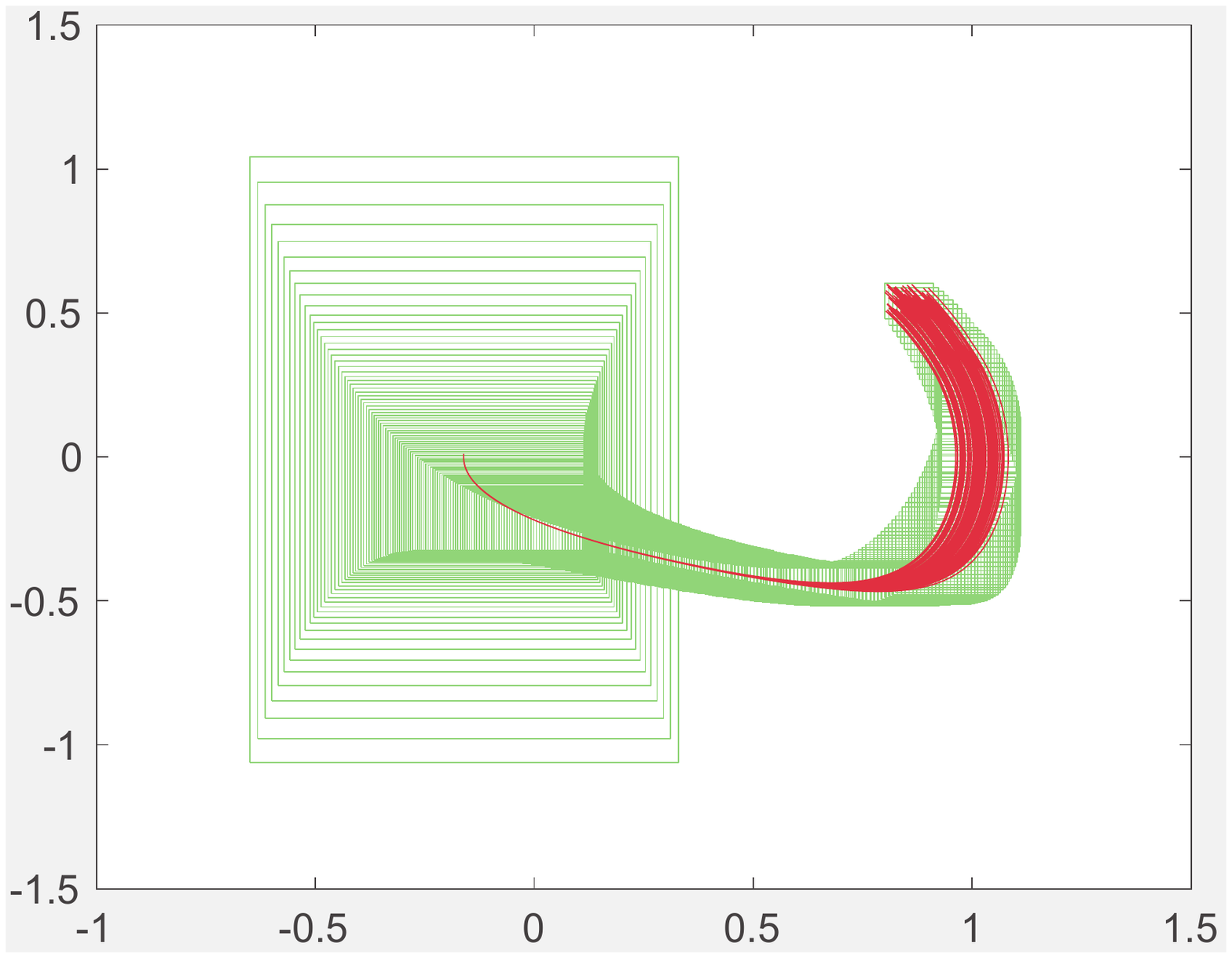}%
		\label{fig:motivating-tanh-interval}%
	}\ \
	\subfloat[][ReachNN]{%
		\includegraphics[width=0.23\textwidth]{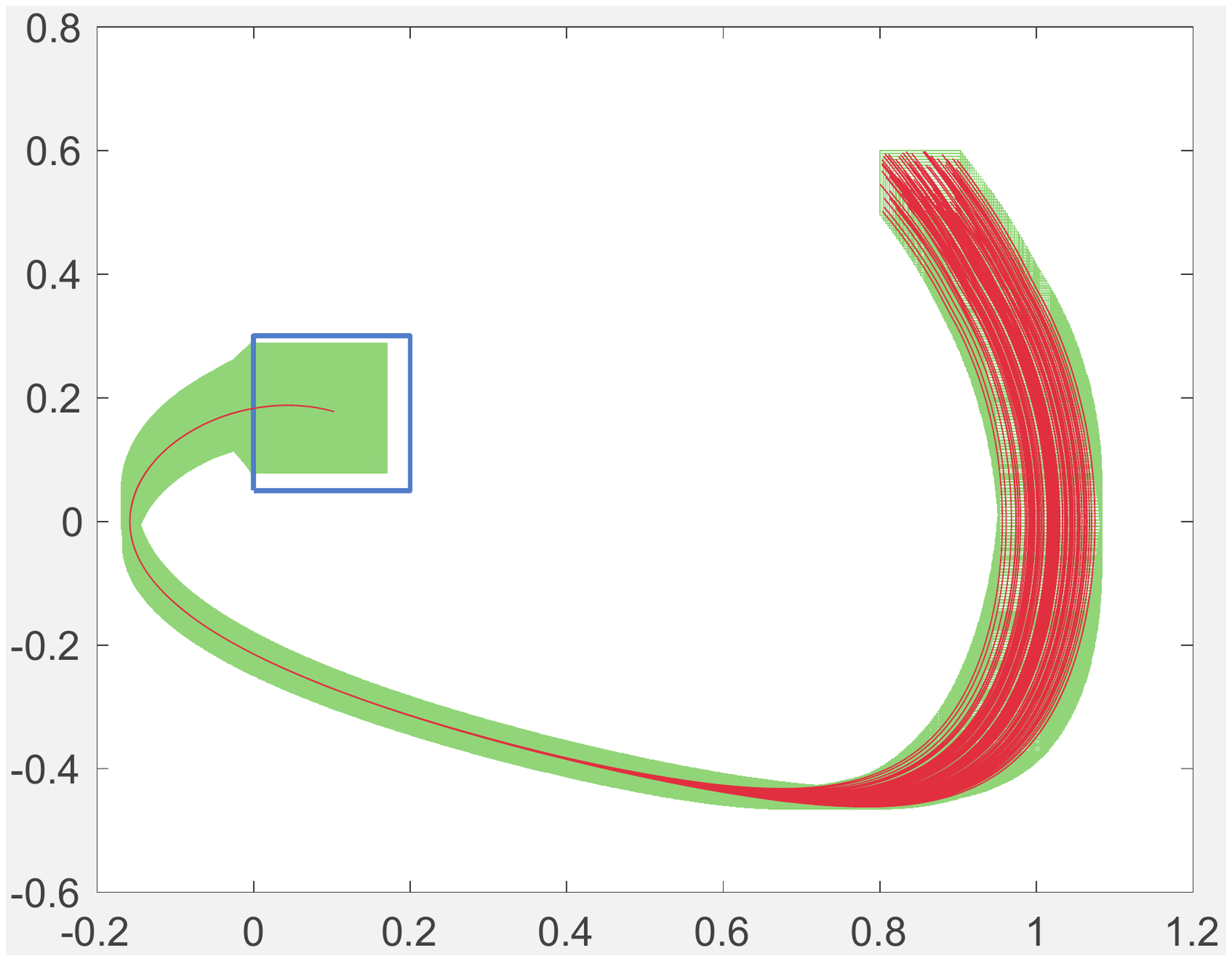}%
		\label{fig:motivating-tanh}%
	}
	\caption{Flowpipes computed by interval overapproximation approach and ReachNN on a NNCS with a heterogeneous neural-network controller that has ReLU and tanh activation functions.}
	\label{fig:motivating}
\end{figure}

\subsection{Benchmarks}

\begin{table*} [htbp] 
\caption{Benchmark setting: For each example $\#$, $ODE$ denotes the ordinary differential equation of the example, $V$ denotes the dimension of the state variable, $\delta_c$ is the discrete control stepsize, $N$ denotes the number of control step, $Init$ denotes the initial state set, $Goal$ denotes the goal set that we hope to verify the system will enter after $N$ steps.}
\small
\setlength\tabcolsep{3.0pt}
\renewcommand{\arraystretch}{1.2}
\begin{tabular}{|c|c|c|c|c|c|}
\hline
\#                 & ODE               & V               & $\delta_c$        &  Init              & Goal  \\ 
\hline
\multirow{1}{*}{1} & \multirow{1}{*}{\minitab[c]{$\dot{x}_1=x_2$, $\dot{x}_2=ux_2^2-x_1$.}} & \multirow{1}{*}{2} & \multirow{1}{*}{0.2} & \multirow{1}{*}{\minitab[c]{$x_1{\in}[0.8,0.9]$, $x_2{\in}[0.5,0.6]$}} & \multirow{1}{*}{\minitab[c]{$x_1{\in}[0,0.2]$, $ x_2{\in}[0.05,0.3]$}} \\ 
\hline
\multirow{1}{*}{2} & \multirow{1}{*}{\minitab[c]{$\dot{x}_1=x_2-x_1^3$, $\dot{x}_2=u$.}} & \multirow{1}{*}{2} & \multirow{1}{*}{0.2} & \multirow{1}{*}{\minitab[c]{$x_1{\in}[0.7,0.9]$, $ x_2{\in}[0.7,0.9]$}} & \multirow{1}{*}{\minitab[c]{$x_1{\in}[-0.3,0.1]$, $ x_2{\in}[-0.35,0.5]$}} \\
\hline
\multirow{1}{*}{3} & \multirow{1}{*}{\minitab[c]{$\dot{x}_1{=}{-}x_1(0.1{+}(x_1{+}x_2)^2)$, $\dot{x}_2{=}(u{+}x_1)(0.1{+}(x_1{+}x_2)^2)$.}} & \multirow{1}{*}{2} & \multirow{1}{*}{0.1} & \multirow{1}{*}{\minitab[c]{$x_1{\in}[0.8,0.9]$, $x_2{\in}[0.4,0.5]$}} & \multirow{1}{*}{\minitab[c]{$x_1{\in}[0.2,0.3]$, $x_2{\in} [-0.3,-0.05]$}} \\ 
\hline
\multirow{1}{*}{4} & \multirow{1}{*}{\minitab[c]{$\dot{x_1}{=}{-}x_1{+}x_2{-}x_3$,  $\dot{x_2}{=}{-}x_1(x_3{+}1)-x_2$, $\dot{x_3}{=}{-}x_1{+}u$}} & \multirow{1}{*}{3} & \multirow{1}{*}{0.1} & \multirow{1}{*}{\minitab[c]{$x_1,x_3{\in} [0.25,0.27]$, $x_2{\in} [0.08,0.1]$}} & \multirow{1}{*}{\minitab[c]{$x_1{\in}[-0.05,0.05]$, $ x_2{\in}[-0.05,0]$}} \\ 
\hline
\multirow{1}{*}{5} & \multirow{1}{*}{\minitab[c]{$\dot{x_1}{=}x_1^3{-}x_2$, $\dot{x_2}{=}x_3$, $\dot{x_3}{=}u$}} & \multirow{1}{*}{3} & \multirow{1}{*}{0.2} & \multirow{1}{*}{\minitab[c]{$x_1{\in}[0.38,0.4]$, $x_2{\in} [0.45,0.47]$,  $ x_3{\in}[0.25,0.27]$}} & \multirow{1}{*}{\minitab[c]{$x_1{\in}[-0.4,-0.28]$, $ x_2{\in}[0.05,0.22]$}} \\ 
\hline
\multirow{2}{*}{6} & \multirow{2}{*}{\minitab[c]{$\dot{x_1} {=} x_2$,\ $
\dot{x_2}{=}{-}x_1{+}0.1\sin(x_3)$, \\ $\dot{x_3}=x_4$, $\dot{x_4}=u$}} & \multirow{2}{*}{4} & \multirow{2}{*}{0.5} & \multirow{2}{*}{\minitab[c]{$x_1{\in}[-0.77, -0.75]$,  $x_2{\in}[-0.45, -0.43]$, \\ $x_3{\in}[0.51, 0.54]$, $x_4{\in}[-0.3, -0.28]$}} & \multirow{2}{*}{\minitab[c]{$x_1 {\in}[-0.1, 0.2]$, $x_2 {\in}[-0.9, -0.6]$}}  \\ 
& & & & & \\
\hline              
\end{tabular}
\label{tab:benchmark}
\end{table*}

\begin{table*} [htbp] 
\caption{Neural network controller setting and experimental results: In each example $\#$, for a neural-network controller, $Act$ denotes the applied activation functions, $S$ and $n$ denote the number of layers and the number of neurons in each hidden layer, respectively. While applying our approach $ReachNN$, $d$ denotes the degree of the polynomial approximation, $\bar{\delta}$ denotes the sampling error precision, $\bar{\varepsilon}$ represents the error bound. For Sherlock, $d$ represents the degree of the Taylor model based approximation. 
Under general settings of Flow* (the order of taylor models are chosen between 5-12 in terms of the degree $d$ in ReachNN, the stepsizes for flowpipe construction are chosen between 1/10-1/100), the verification results for each approach are evaluated by two metrics: Boolean $ifReach$ following by a number indicates whether this approach verifies the reachability or not: $Yes(n)$ means the system is proven to reach the goal set after $n$ steps, while $Unknown(n)$ means the reachable set at any step $k\leq n$ is not a subset of the goal set and flowpipes start to blow up after $n$ steps. The computation duration $time$ (in seconds) indicates the efficiency. We use "--" to represent that a certain approach is not applicable.}\label{tab:experiment}
\small
\setlength\tabcolsep{5.8pt}
\renewcommand{\arraystretch}{1.2}
\begin{tabular}{|c|c|c|c|c|c|c|c|c|c|c|c|c|c|c|}
	\hline
	\multirow{2}{*}{\#}                                                                                                            & \multicolumn{3}{c|}{NN Controller} & \multicolumn{5}{c|}{ReachNN}  & \multicolumn{3}{c|}{Sherlock\cite{Dutta_Others__2019__Reachability}}  & \multicolumn{2}{c|}{Verisig}      \\ \cline{2-14}
	&        Act        & layers    & n   & d   & $\bar{\delta}$       & $\bar{\varepsilon}$ & ifReach & time & d       & ifReach & time & ifReach & time \\ \hline
	\multirow{4}{*}{1} & ReLU       &   3   &  20   &    $[1,1]$  & 0.001  &      0.0009995                    &    Yes(35)     &    3184    & 2 & Yes(35) & 41 & --  & -- \\ \cline{2-14} 
	& sigmoid    &   3   &  20   &   $[3,3]$   & 0.001  &       $0.0077157$       &  Yes(35)       &    779     & -- & -- & -- & Unknown(22)  &  --  \\ \cline{2-14} 
	& tanh       &   3   &  20   &   $[3,3]$   &  0.005 &      $0.0117355$    &   Unknown(35)      &     --    & -- & -- & -- & Unknown(22) &   --  \\ \cline{2-14} 
	& ReLU+tanh       &   3   &  20   &   $[3,3]$  & 0.01   &      $0.0150897$    &   Yes(35)      &    589     & -- & --  & -- & --  & --   \\ \hline
	\multirow{4}{*}{2} & ReLU       &   3   &  20   &    $[1,1]$ &  0.01  &       $0.0090560$                     &    Yes(9)     &    128    & 2 & Yes(9) & 3 & -- &   --   \\ \cline{2-14} 
	& sigmoid    &   3   &  20   &   $[3,3]$  & 0.01   &     $0.0200472$       &   Yes(9)   &    280     & --  & -- & -- & Unknown(7) & --  \\ \cline{2-14} 
	& tanh       &   3   &  20   &   $[3,3]$  &  0.01  &     $0.0194142$      &   Unknown(7)   &    --     & -- & -- & --   &  Unknown(7)  &  --   \\
	\cline{2-14} 
	& ReLU+tanh       &   3   &  20   &   $[3,3]$  & 0.001   &      $0.0214964$    &   Yes(9)      &    543     & -- & -- & -- & -- & -- \\ \hline 
	\multirow{4}{*}{3} & ReLU       &   3   &  20   &    $[1,1]$     &  0.01 &    $0.0205432$              &        Yes(60)        &   982     & 2 & Yes(60) & 139 & --  & --    \\ \cline{2-14} 
	& sigmoid    &   3   &  20   &   $[3,3]$  & 0.005  &  $0.0060632$         &    Yes(60)     &    1467    & -- & -- & -- & Yes(60) &   27   \\ \cline{2-14} 
	& tanh       &   3   &  20   &   $[3,3]$  &  0.01 &  $0.0072984$      &   Yes(60)     &   1481     & -- & -- & -- &Yes(60)  & 26     \\ \cline{2-14} 
	& ReLU+tanh       &   3   &  20   &   $[3,3]$   & 0.01  &      $0.0230050$    &   Unknown(60)      &    --     & -- & -- & -- & -- & --\\   \hline
	\multirow{4}{*}{4} & ReLU       &  3    &  20   &    $[1,1,1]$  &  0.005   &               $0.0048965$       &  Yes(5)  &   396     & 2 & Yes(5) & 19 & -- & --\\                                 \cline{2-14}
	& sigmoid    &  3    &   20  &    $[2,2,2]$     &  0.01    & $0.0096400$                     &  Yes(10)      &     253     & -- & -- & -- & Yes(10) & 7 \\ \cline{2-14} 
	& tanh    &  3    &   20  &    $[2,2,2]$     &  0.01    & $0.0095897$                     &  Yes(10)      &     244     & -- & -- & -- & Yes(10) & 7 \\ \cline{2-14} 
	& ReLU+sigmoid       &   3   &  20   &   $[2,2,2]$ &  0.01  &      $0.0096322$    &   Yes(5)      &    108     & -- & -- & -- & -- & -- \\   \hline
	\multirow{4}{*}{5} & ReLU       &   4   &  100   &    $[1,1,1]$     &  0.004    & $0.0039809$       &    Yes(10)     &    5487      & 2 & Yes(10) & 12 & -- & --   \\ \cline{2-14}
	& sigmoid    &  4    &   100  &    $[2,2,2]$     & 0.004     &   $0.0039269$      &   No(10)    &    8842      & -- & -- & -- & Unknown(10) & -- \\ \cline{2-14} 
	& tanh       &   4   &  100   &  $[2,2,2]$   & 0.004  &      $0.0038905$  &  Unknown(10)      &  7051      & -- & -- & -- & Unknown(10)  &   --   \\ \cline{2-14}
	& ReLU+tanh       &   4   &  100   & $[2,2,2]$   & 0.04   &   $0.0039028$       &   Unknown(10)    &      7369    & -- & -- & -- & -- & -- \\
	\hline
	\multirow{4}{*}{6} & ReLU       &    4  &  20   &    $[1, 1, 1 ,1]$  &  0.001   &       $0.0096789$                 &   Yes(10)      &    7842     & 2 & Yes(10) & 33 & -- & --   \\ \cline{2-14}
	& sigmoid    &  4    &   20  &    $[1,1,1,1]$     &  0.001    &         $0.0082784$             &   No(7)    & 32499       & -- & -- & -- & Yes(10) & 34 \\ \cline{2-14} 
	& tanh       &   4   &  20   &  $[1,1,1,1]$   & 0.001  &     $0.0156596$   &   No(7)    &    3683    & -- & -- & -- & Yes(10) &   35   \\ \cline{2-14}
	& ReLU+tanh       &   4   &  20   &   $[1,1,1,1]$      &  0.001   & $0.0091648$    &   Yes(10)     &   10032     & -- & -- & -- & -- & --\\   \hline
\end{tabular}
\label{tab:results}
\end{table*}

Table~\ref{tab:benchmark} shows the benchmarks settings, while Table~\ref{tab:results} shows the experiment results of ReachNN, Sherlock \cite{Dutta_Others__2019__Reachability} and Verisig \cite{ivanov2018verisig}. We can first find that our approach can verify most examples, with a few exceptions. The simulation trajectories along with overapproximate reachable set computed by ReachNN, Sherlock and Verisig of some selected examples are shown in Figure \ref{fig:simulation}. The reason why we fail in these examples are mainly due to the relatively large approximation error. As shown later in Section \ref{sec:challenges}, an interesting phenomenon is that Bernstein polynomial based approach may perform differently with respect to the type of activation functions, which we will consider in our future work.

Benefiting from its generality, our approach can handle all these examples, while Sherlock and Verisig are only applicable to a few of them. Furthermore, the heterogeneous neural networks that contain multiple types of activation functions, which is common in practice~\cite{beer1989heterogeneous,Lillicrap2016ContinuousCW}, can only be handled by our approach. However we acknowledge that our method costs much more time than Sherlock and Verisig (see Table~\ref{tab:results}). The main reason is the large number of samples needed in estimating the error for a Bernstein polynomial. Since we only require a neural network to be Lipschitz continuous, the estimation of approximation error is quite conservative. However, such limitation may be overcome by considering more information of the activation functions. We plan to explore it in our future work.

\begin{figure*}[htb]
	\centering	
	\subfloat[][Ex1-tanh]{%
		\includegraphics[width=0.24\textwidth]{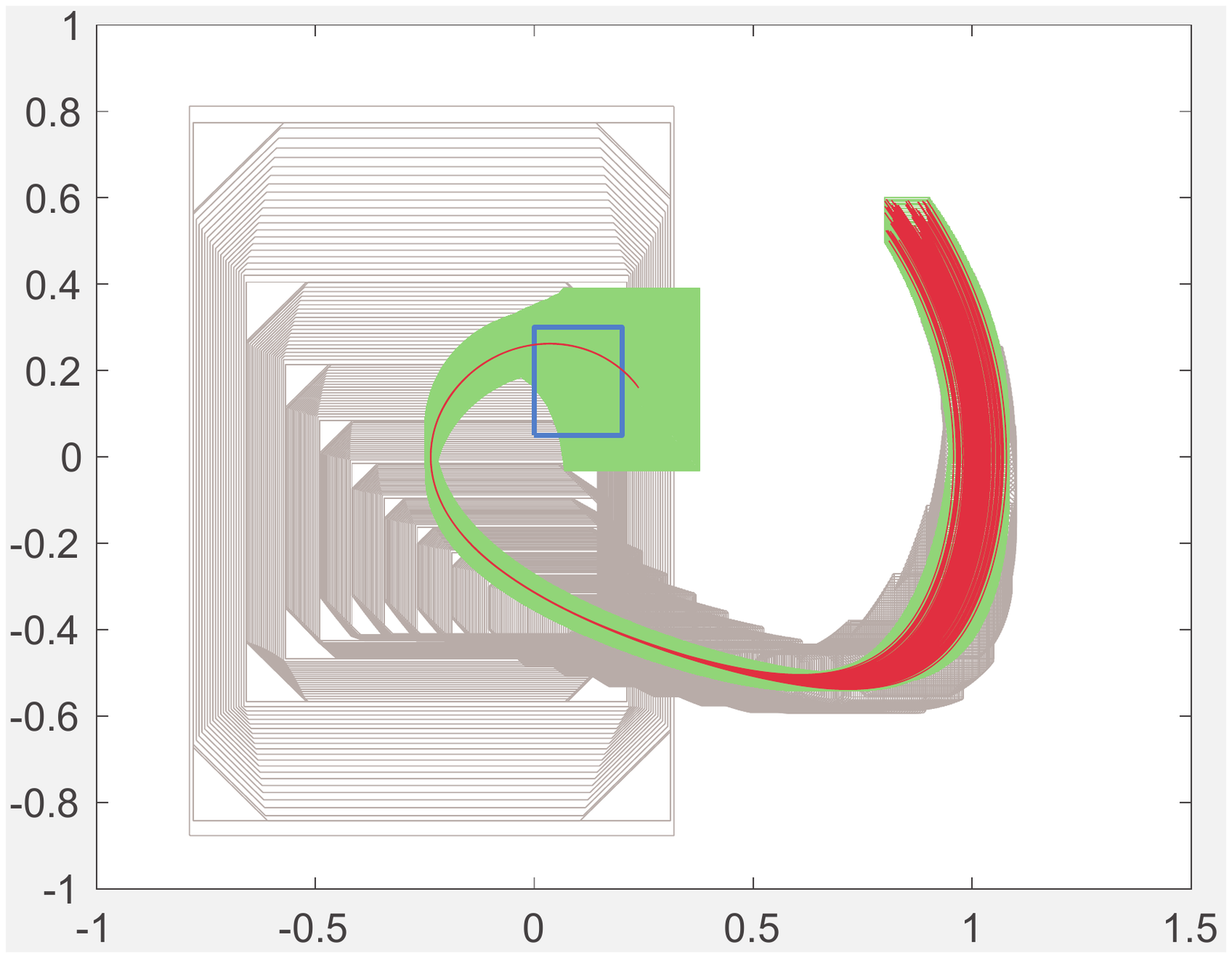}%
		\label{fig:ex3-relu}%
	}\ \
	\subfloat[][Ex2-sigmoid]{%
		\includegraphics[width=0.24\textwidth]{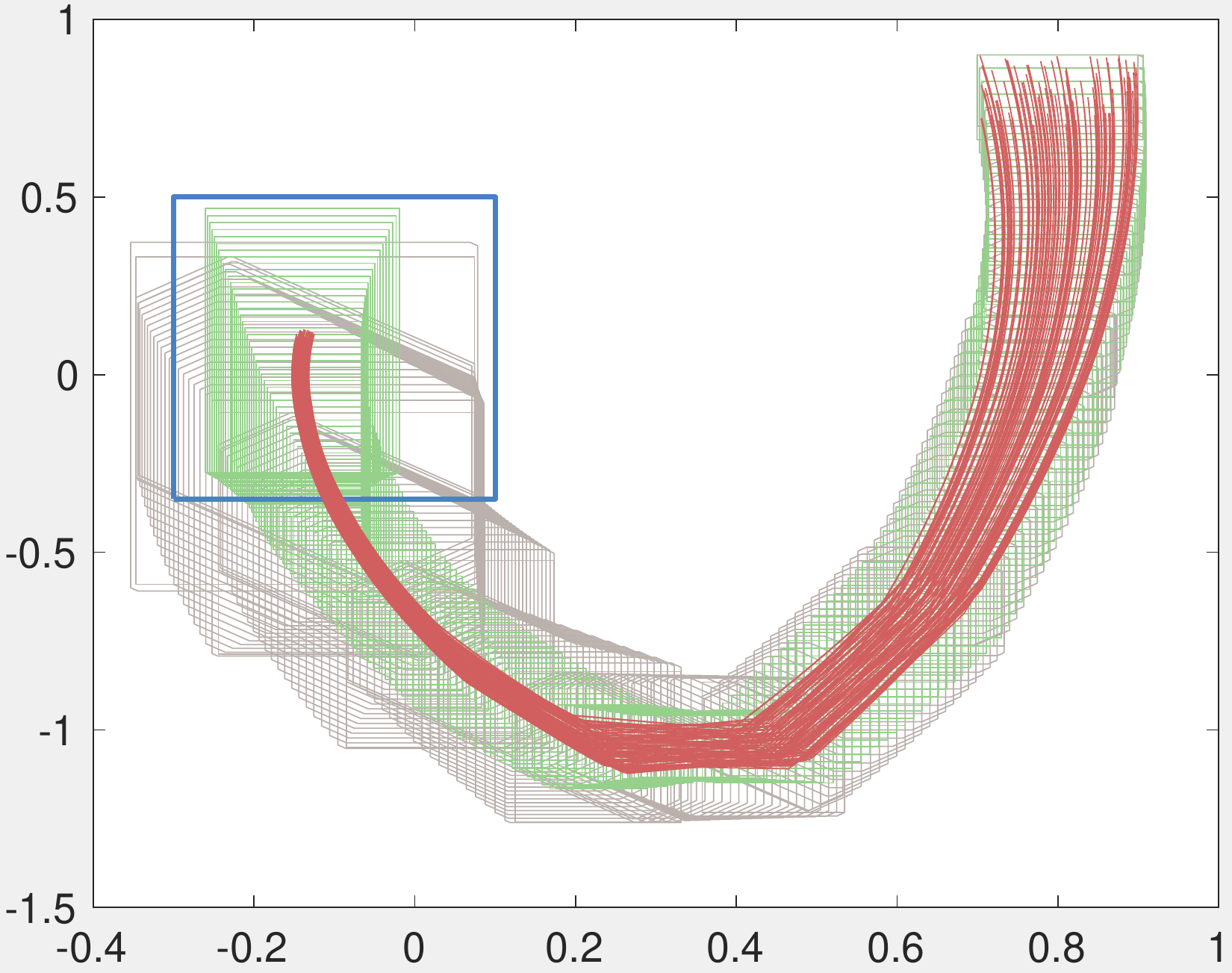}%
		\label{fig:ex6-relu}%
	}\ \
	\subfloat[][Ex4-sigmoid]{%
    	\includegraphics[width=0.24\textwidth]{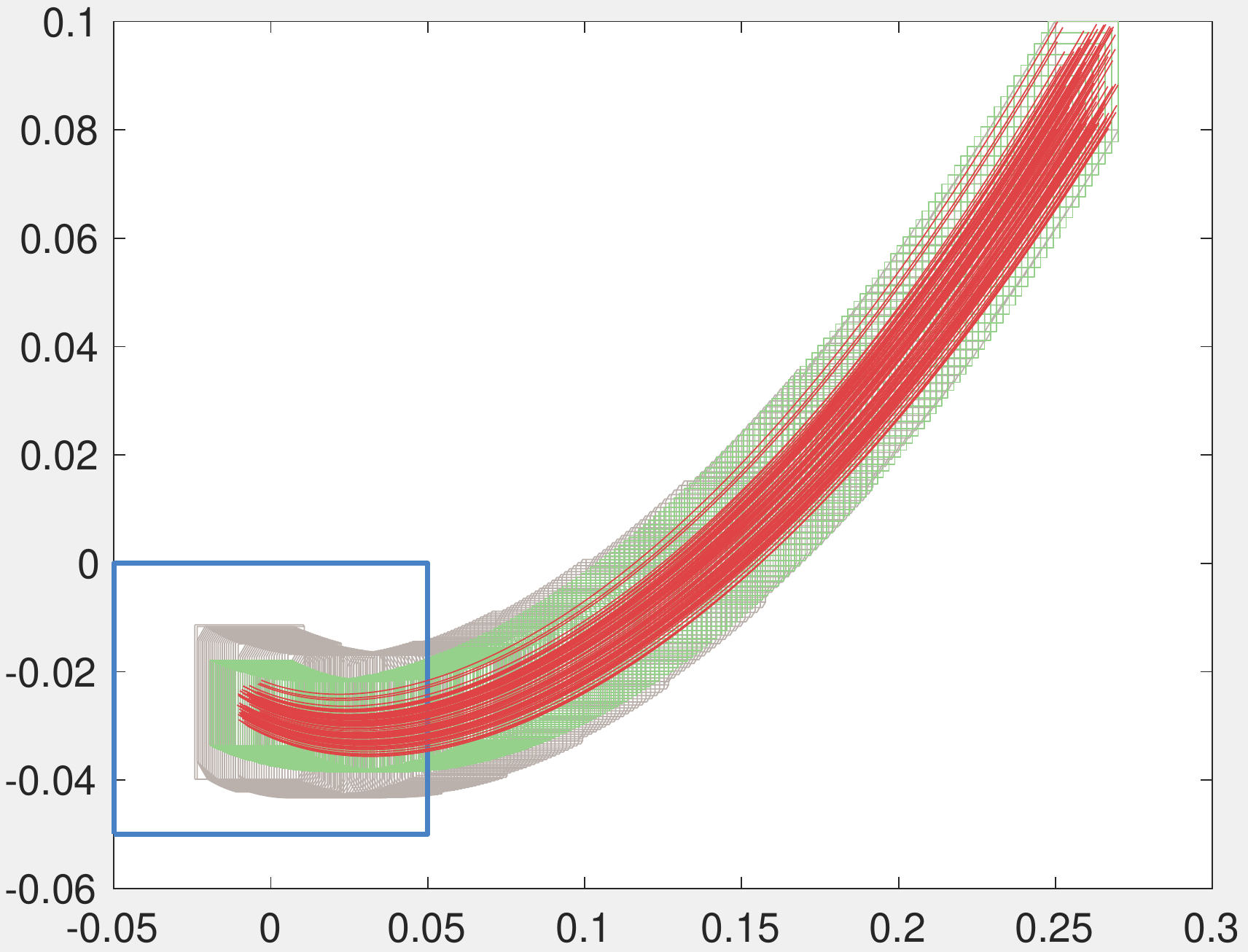}%
    	\label{fig:ex8-relu}%
	}\ \
	\subfloat[][Ex5-ReLU]{%
    	\includegraphics[width=0.24\textwidth]{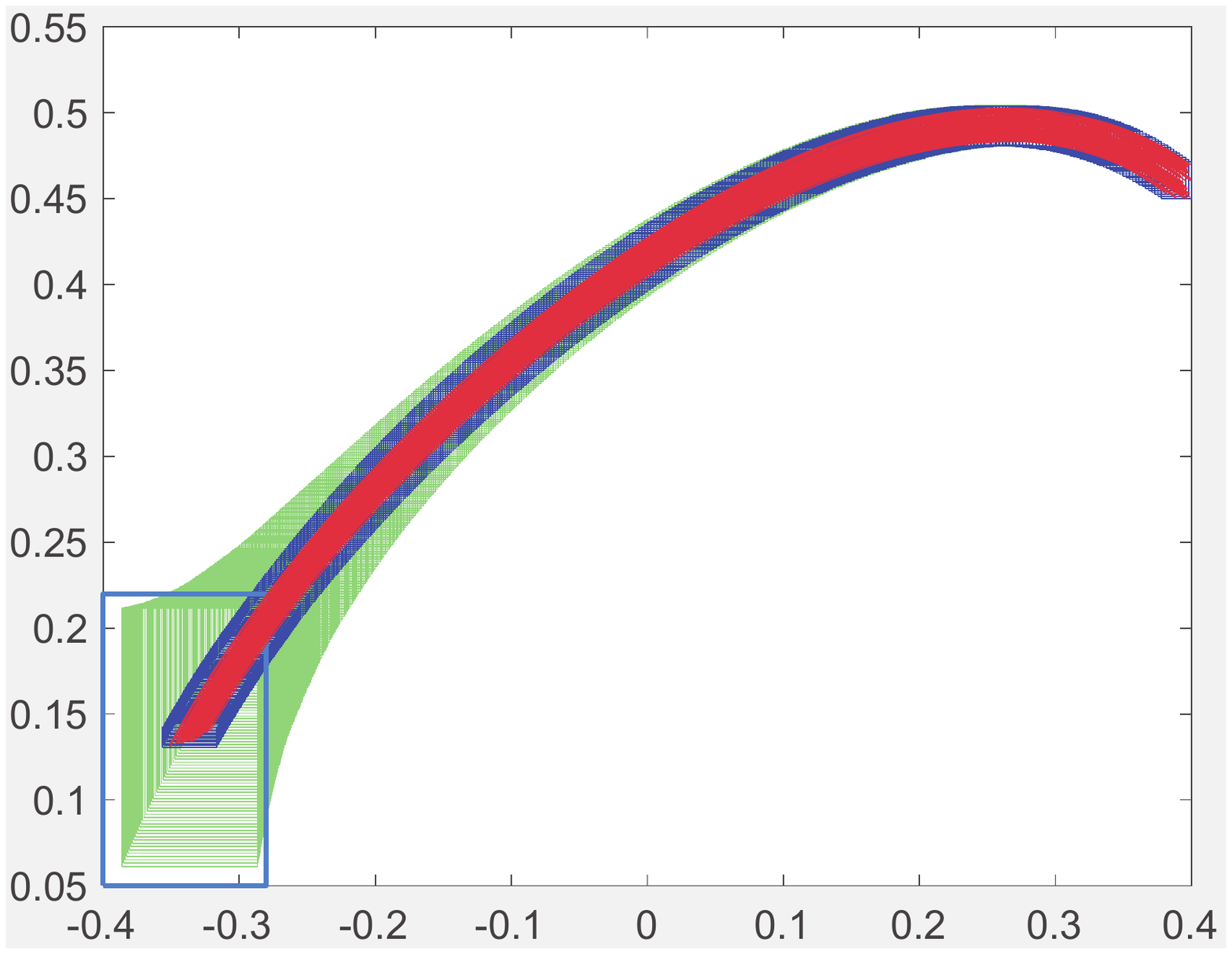}%
    	\label{fig:ex9-relu}%
	}\\
	\subfloat[][Ex1-sigmoid]{%
		\includegraphics[width=0.24\textwidth]{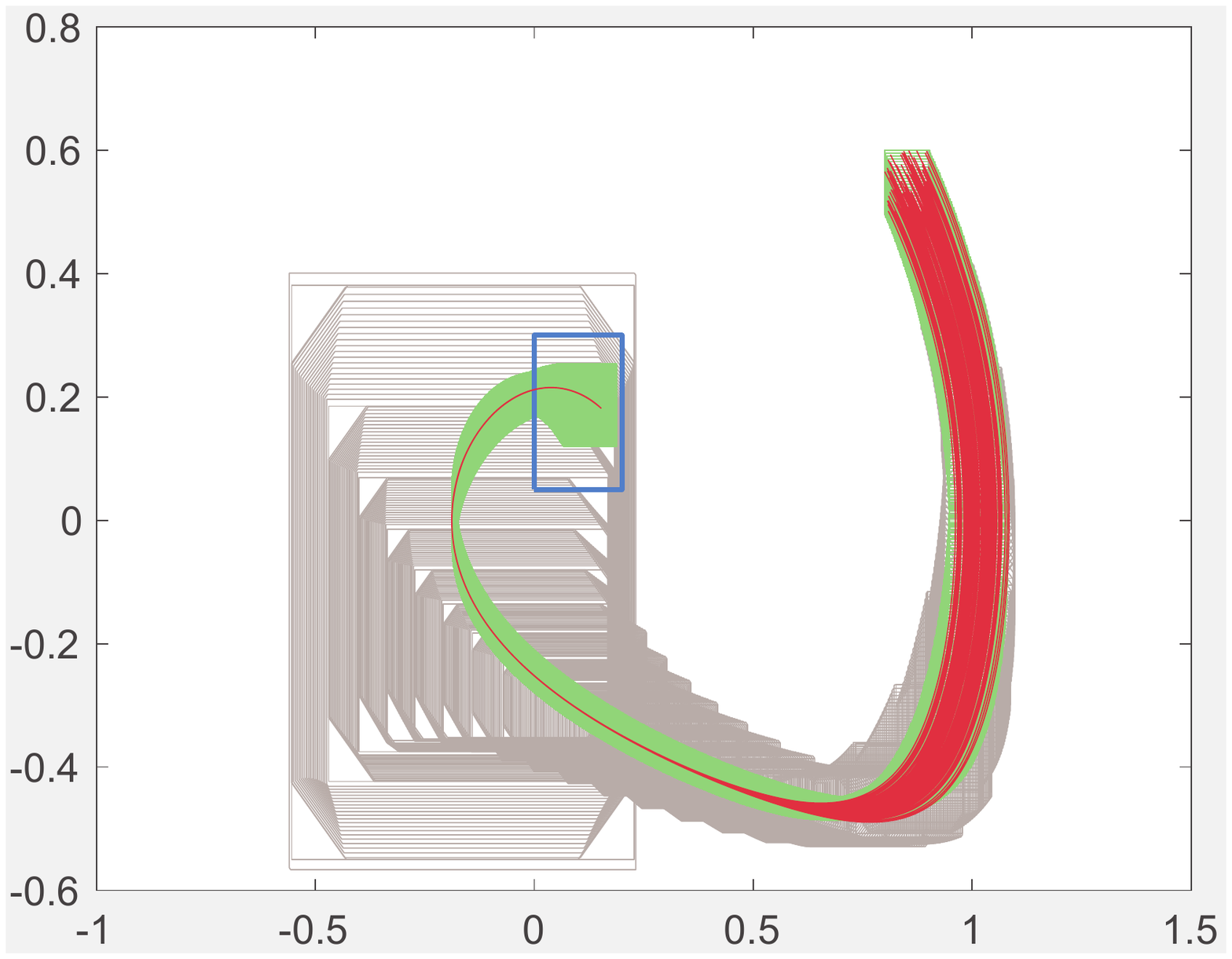}%
		\label{fig:ex3-sigmoid}%
	}\ \
	\subfloat[][Ex2-ReLU-tanh]{%
		\includegraphics[width=0.24\textwidth]{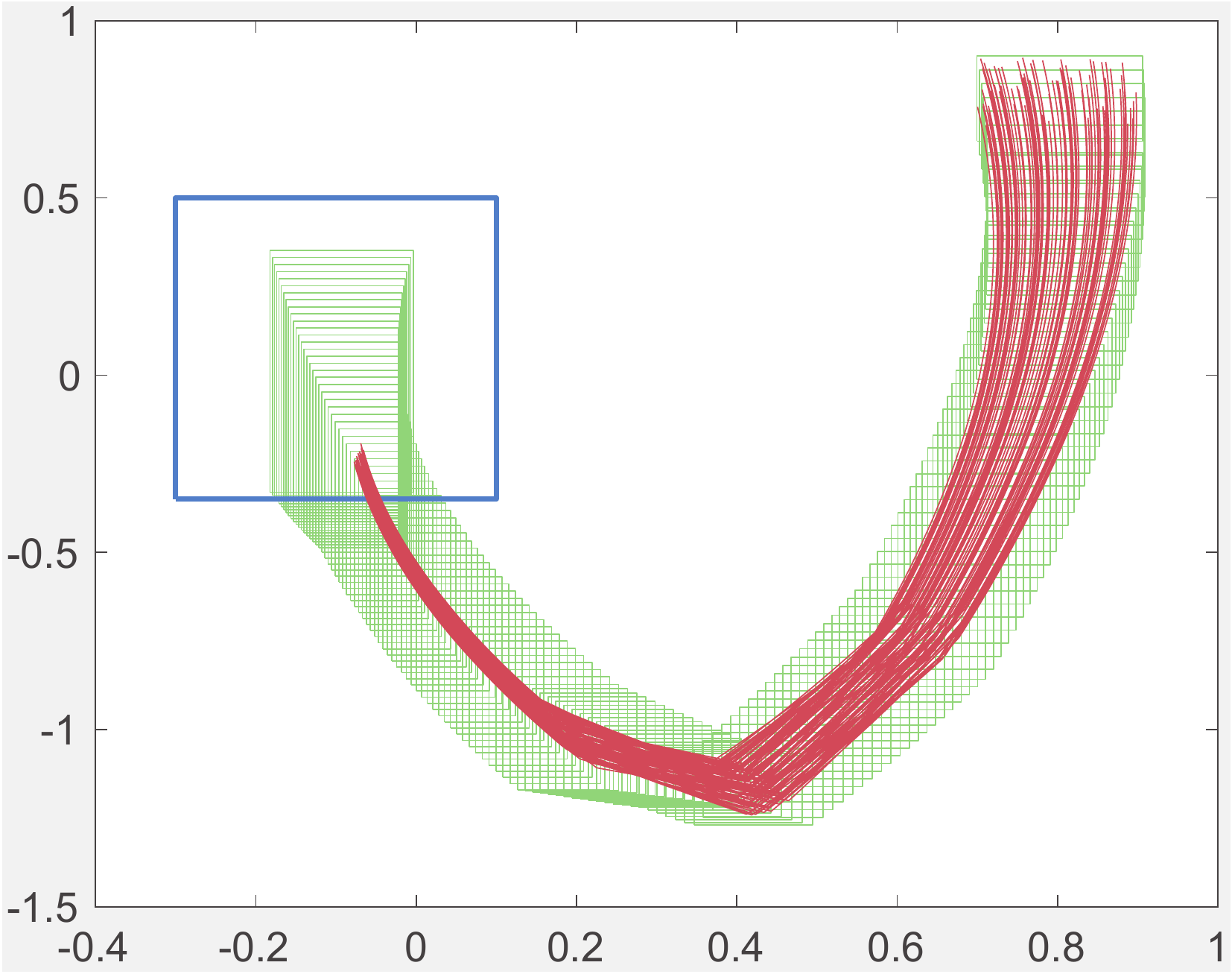}%
		\label{fig:ex6-sigmoid}%
	}\ \
	\subfloat[][Ex4-ReLU-tanh]{%
		\includegraphics[width=0.24\textwidth]{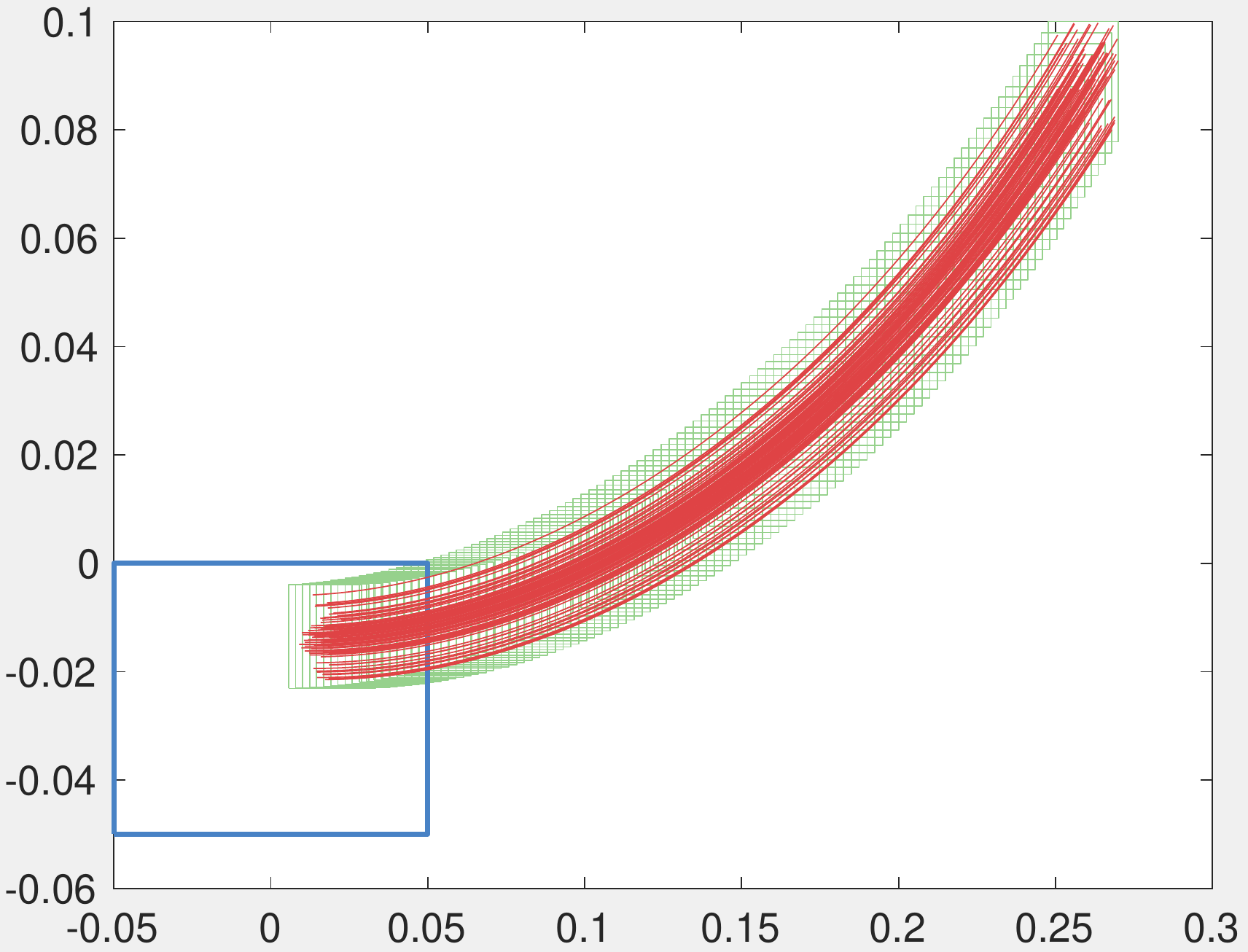}%
		\label{fig:ex7-relu}%
	}\ \
	\subfloat[][Ex6-ReLU-tanh]{%
		\includegraphics[width=0.24\textwidth]{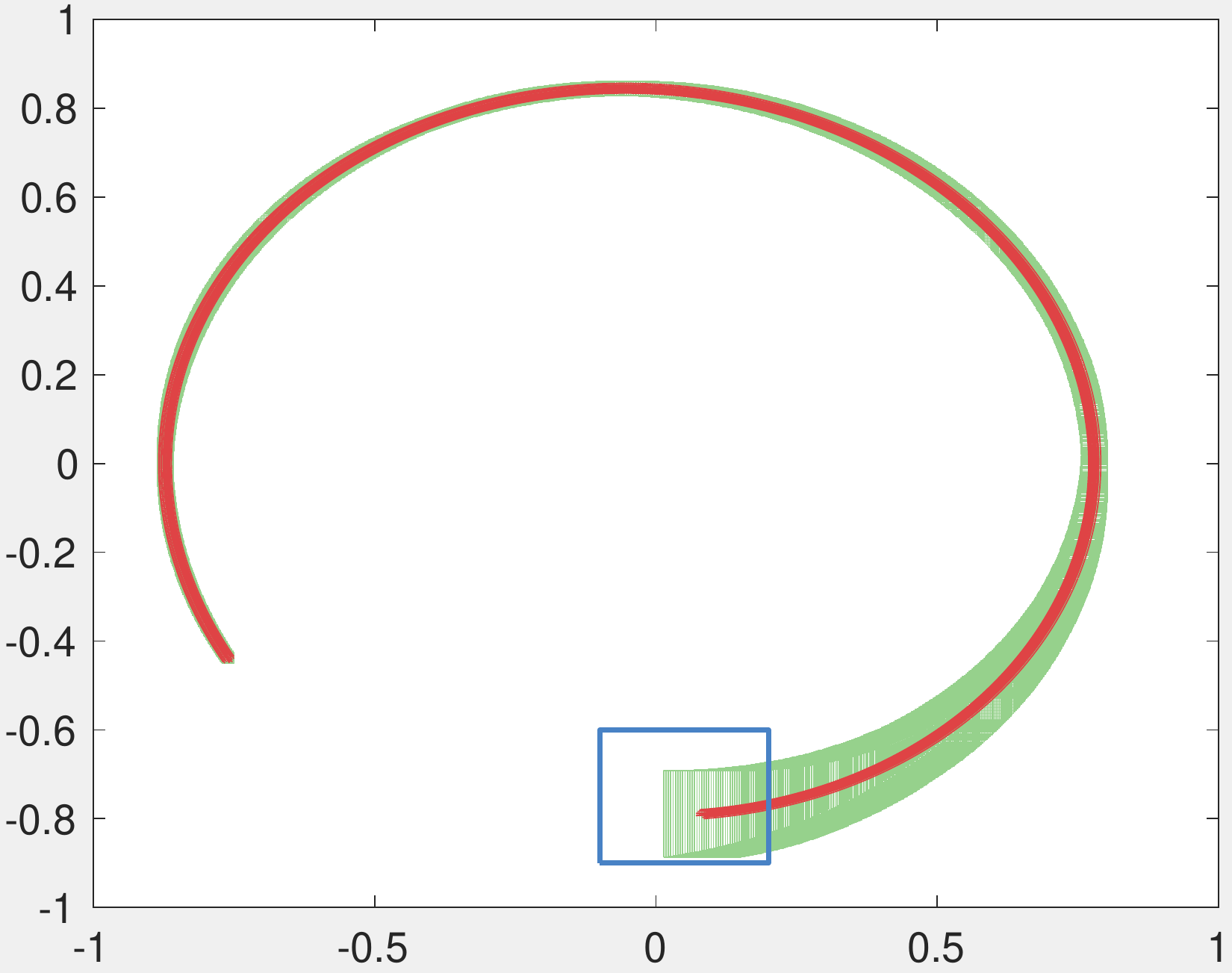}%
		\label{fig:ex8-relu-tanh}%
	}
	\caption{Flowpipes for the selected examples: Red curves denote the trajectories of $x_1$ and $x_2$ of the system simulated from sampled states within the initial set. Green rectangles are the constructed flowpipes as the overapproximation of the reachable state set by our approach, gray rectangles are the flowpipes computed based on Verisig, and Sherlock computes the flowpipes represented as deep blue rectangles. The blue rectangle is the goal area. Neither Verisig nor Sherlock can analyze the networks in (f), (g) or (h) due to the presence of both ReLU and tanh activation functions.}
	\label{fig:simulation}
\end{figure*}

\section{Discussion and Open Challenges} \label{sec:challenges}
In this section, we will show further insights into our approach and discuss some remaining challenges. 

\smallskip
\noindent

\textbf{ReLU v.s. tanh v.s. sigmoid: Approximation performance of Bernstein polynomials.}
We empirically explore the approximation performance of Bernstein polynomials for different neural networks by sampling. We take Example 2 over $X=[0.7,0.9]$ in the benchmark for instance. For each neural network controller, we sample a large number of points and plot the function value of the controller and its approximation  over $X=[0.7,0.9]$ (Figure \ref{fig:approximation}). The orders of magnitude of the error for ReLU network, sigmoid network and tanh network are $10^{-12}$, $10^{-4}$, $10^{-4}$, respectively. First, the approximation errors for all these three networks are fairly small for the $0.2\times 0.2$ box, which indicates that Bernstein polynomial based approach is promising if more efficient error estimation approaches could be designed. Secondly, we can see that Bernstein polynomials can achieve a higher approximation precision for the ReLU network than the other two. This motivates us to further explore the impact of the inner structure of different neural networks on the approximation performance in future work.

\begin{figure}[tbp]
	\centering	
	\subfloat[][Approximation for ReLU neural network]{%
		\includegraphics[width=0.23\textwidth]{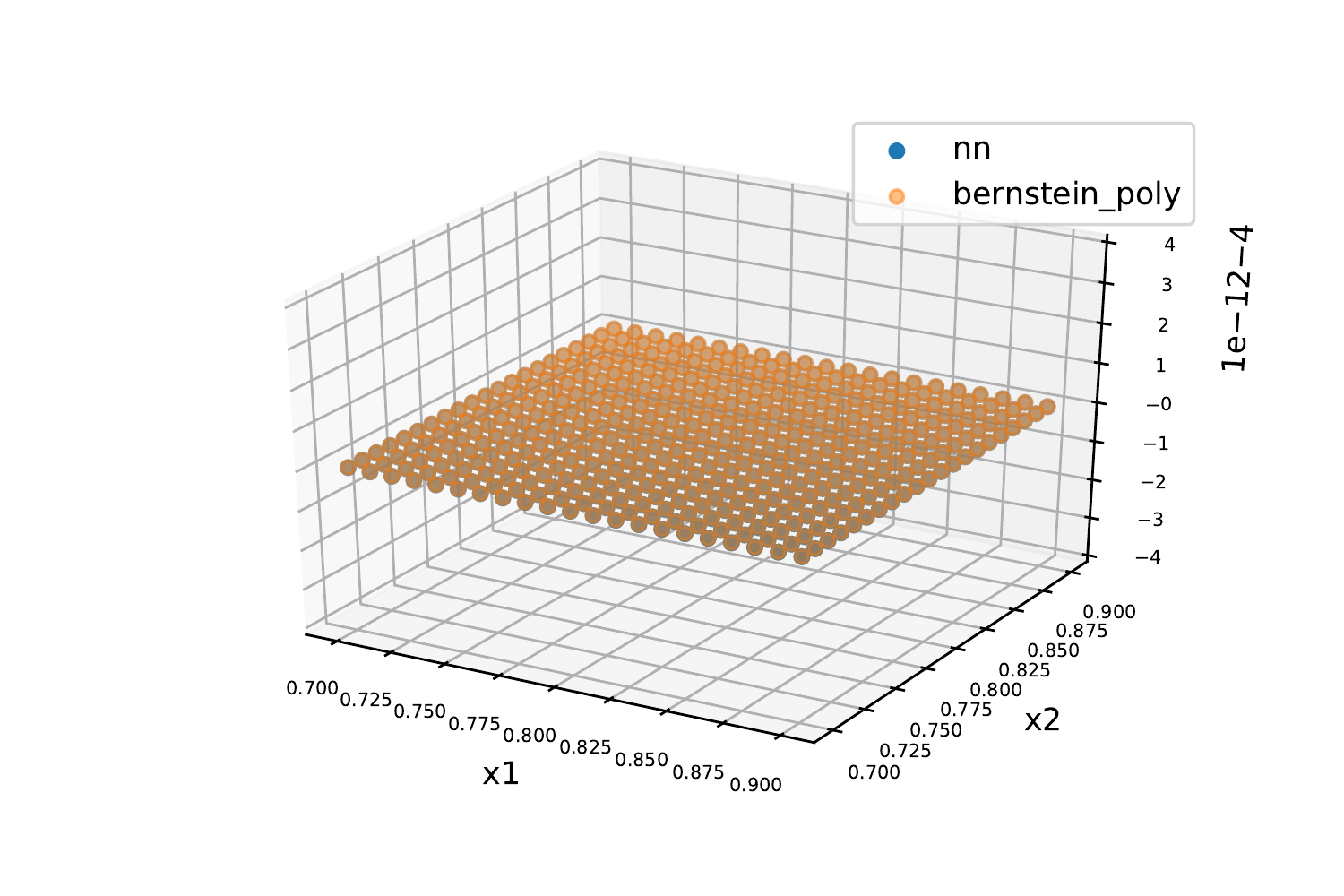}%
		\label{fig:bernsp_relu}%
	}\\
	\subfloat[][Approximation for sigmoid neural network]{%
		\includegraphics[width=0.23\textwidth]{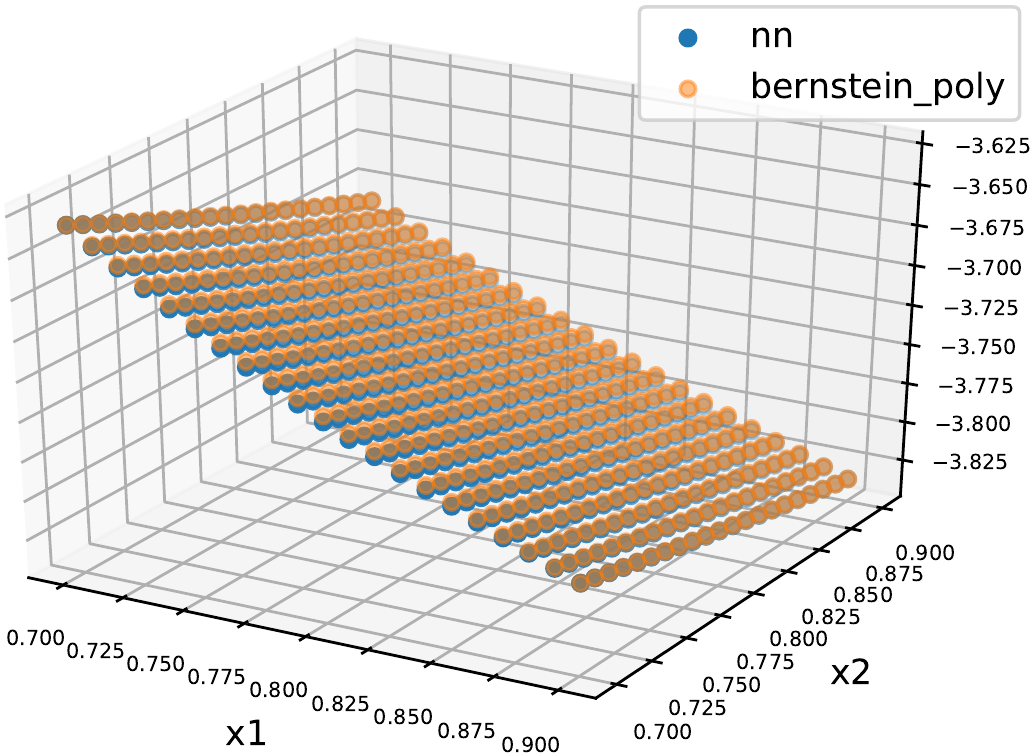}%
		\label{fig:bernsp_sigmoid}%
	}\ 
	\subfloat[][Approximation for tanh neural network]{%
		\includegraphics[width=0.23\textwidth]{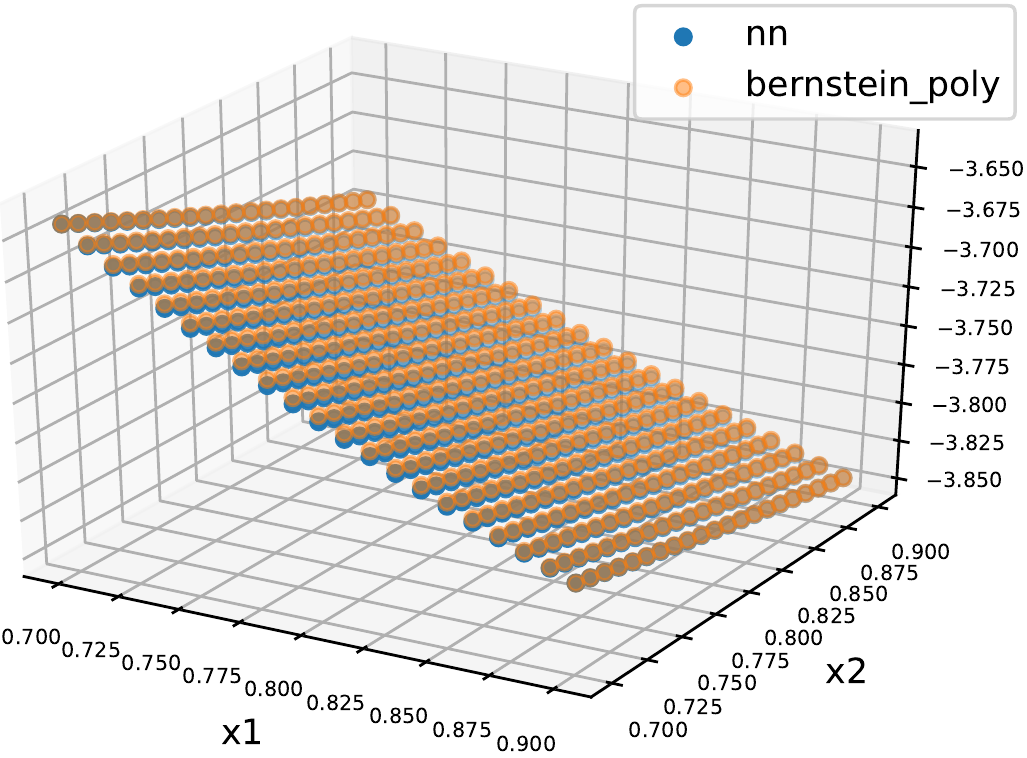}%
		\label{fig:bernsp_tanh}%
	}
	\caption{Bernstein polynomial based approximation for different neural networks. In Figure~\ref{fig:bernsp_relu}, \ref{fig:bernsp_sigmoid}, \ref{fig:bernsp_tanh}, $x1$-axis and $x2$-axis are $x_1$ and $x_2$ respectively, while $z$ axis is the value of neural network/polynomial approximation. The blue point and yellow point are the sample values of the neural network and the approximation polynomial respectively.}
	\label{fig:approximation}
\end{figure}

\smallskip
\noindent
\textbf{Small Lipschitz constant v.s. Large Lipschitz constant.} Given that the approximation error of Bernstein polynomials is upper bounded linearly by the Lipschitz constant of neural network, we conducted preliminary study on the impact of the Lipschitz constant on ReachNN. We consider the dynamical system of an inverted pendulum on a cart:
\begin{align*}
    \dot{x}_1&=x_2, \\
    \dot{x}_2&=\frac{-m g \sin{(x_3)} \cos{(x_3)}+m l x_4^2 \sin{(x_3)} + f m x_4 \cos{(x_3)} + u}{M+(1-\cos{(x_3)}^2)m}, \\
    \dot{x}_3&=x_4, \\
    \dot{x}_4&= \frac{(M + m) * (g \sin{(x_3)} - f x_4) - (l m x_4^2 \sin{(x_3)} + u)}{l (M + 1 - \cos{x_3}^2)m},
\end{align*}
where the angular position and velocity of the pendulum are $x_1$ and $x_2$, the position and velocity of the cart are $x_3$ and $x_4$, the pendulum mass is $m=0.23$, the cart mass is $M=2.4$, gravitational acceleration is $g=9.8$, the length of pendulum is $l=0.36$, and the coefficient of friction is $f=0.1$. The goal is to stabilize the pendulum at the upward position and verify whether the cart position remains in $[2, 4]$ after $25$ control steps. The system will start randomly from $x_1 \in [0.5, 0.55],x_2 \in [-1, -0.95],x_3 \in [2.5, 2.55],x_4 \in [0, 0.05]$.

Given a trained five-layer ReLU neural network $\kappa_1$ with the width of each layer as $[100, 1, 2, 1, 2]$, the computed Lipschitz constant is $874.5$~\cite{ruan2018reachability}. Although this Lipschitz constant is an upper bound of the best Lipschitz constant, by using the same Lipschitz constant computation method, we can use it as a proxy to estimate the differences between the best Lipschitz constants of different neural networks. 
Figure~\ref{fig:ex10-tanh-interval}
illustrates that the Lipschitz constant of a network has a significant impact on our Bernstein polynomial-based reachablility analysis. For $\kappa_1$ as shown in Figure~\ref{fig:ex10-tanh-interval}, the error bound estimation grows rapidly and end up becoming too large at the $10$th step. The blue curves are continuations of the simulation results from $10$ to $20$ steps after reachability analysis becomes unreliable.

\smallskip
\noindent
\textbf{Neural Network Retraining.} Thus, to handle neural-network controllers with large Lipschitz constants, we retrained the network by sampling input-output data from the original network and added a penalty term for Lipschitz constant in the loss function. To reduce the Lipschitz constant, $L_\theta$, of the retrained network $\kappa'(x;\theta)$, we consider the following empirical risk minimization problem:
\begin{equation}
    \min_\theta \frac{1}{N} \sum_{i=1}^{N}L(\kappa'(x;\theta), \kappa(x)) + \lambda L_\theta,
\end{equation}
where $\lambda \in \mathbb{R}_+$ is a regularization factor and $L_\theta = \prod_{l=1}^L \Vert W_\theta^l \Vert_2$ as the Lipschitz constant~\cite{ruan2018reachability}. We refer to the second term as the \emph{Lipschitz constant regularizer}. 
The gradient of the \emph{Lipschitz constant regularizer} is only related to the largest singular value and corresponding vectors of $W_\theta^l$ and projected by weights of other layers. This means each retrained weight matrix, $W_\theta^l$, does not shrink significantly in the directions orthogonal to the first right singular vector and preserves potential important information in the original network. 
We use classical gradient descent method to do the optimization. A similar approach is also mentioned in~\cite{yoshida2017spectral}.

For the inverted pendulum on a cart example, we retrained a three-layer ReLU neural network $\kappa_2$ with 50 neurons each layer, which has the Lipschitz constant upper bound of $14.7$. In Figure~\ref{fig:ex10-stability}, we show state evolution of angular position and angular velocity controlled by the original neural network $\kappa_1$ and by the retrained neural network $\kappa_2$, respectively. In Figure~\ref{fig:ex10-tanh-interval} and Figure~\ref{fig:ex10-tanh}, we show that the retrained neural network can produce results comparable to the original neural network. For the neural network with a small Lipschitz constant, we postulate that Bernstein polynomials can track the behaviors of the neural network better. In addition, ReachNN provides tighter error bound estimation. As a result, the overapproximation quality is significantly improved.
This idea is aligned with \textit{model compression}~\cite{bucilua2006model} (or \textit{distillation}~\cite{Hinton2015DistillingTK}), which uses high-performance neural networks to guide the training of shallower~\cite{ba2014deep} or more structured neural networks. One key observation in \cite{ba2014deep} is that deep and complex neural networks perform better not because of better representation power, but because they are better regularized and therefore easier to train. Here, we consider the retraining process as a form of regularization that effectively maintains the performance of the original network but obtains a smaller Lipschitz constant.
We plan to explore the tighter connection between training and verification more thoroughly in future work.

\begin{figure}[tbp]
	\centering	
	\subfloat[][Stability performance comparison]{%
		\includegraphics[width=0.23\textwidth]{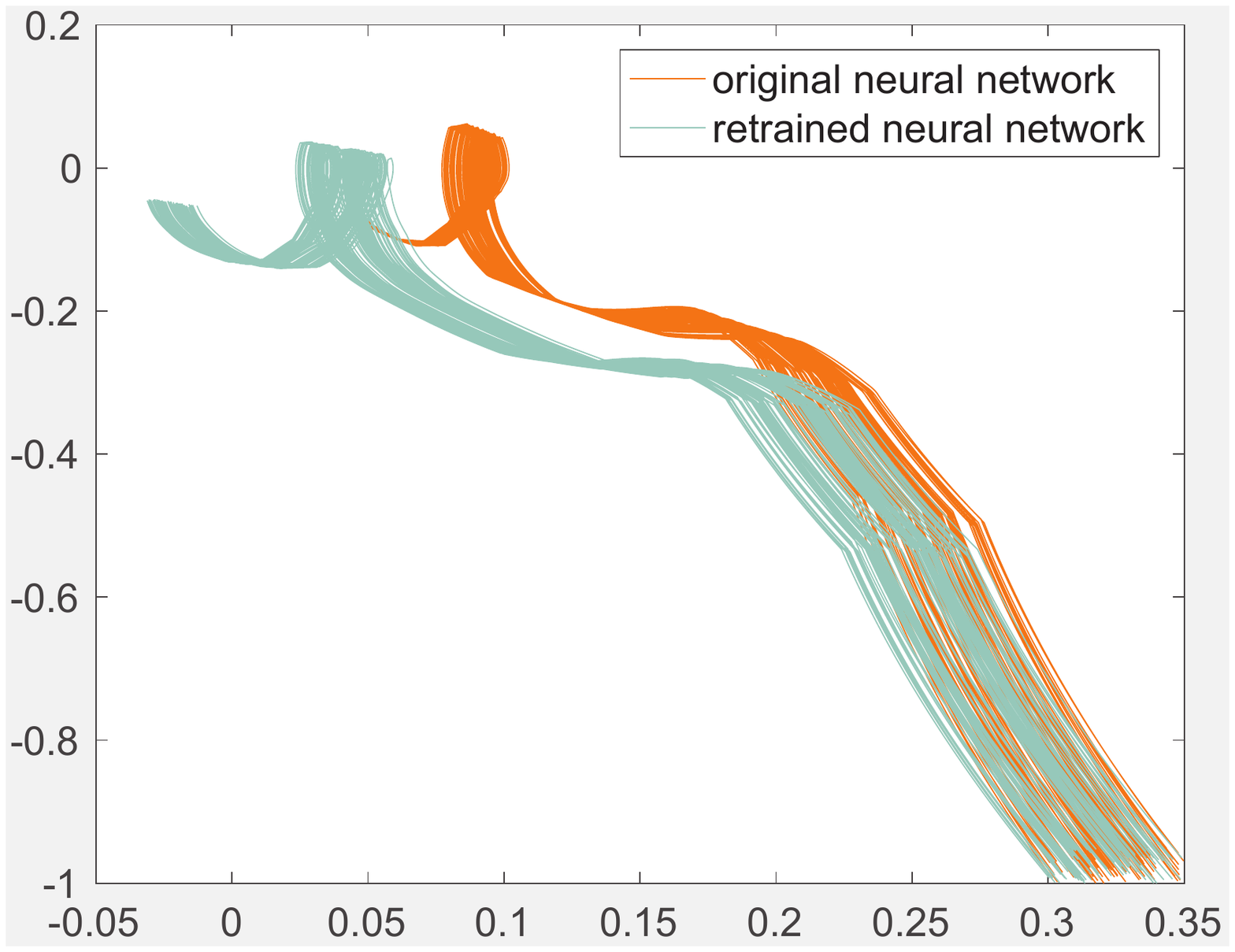}%
		\label{fig:ex10-stability}%
	}\\
	\subfloat[][Large Lipschitz constant]{%
		\includegraphics[width=0.23\textwidth]{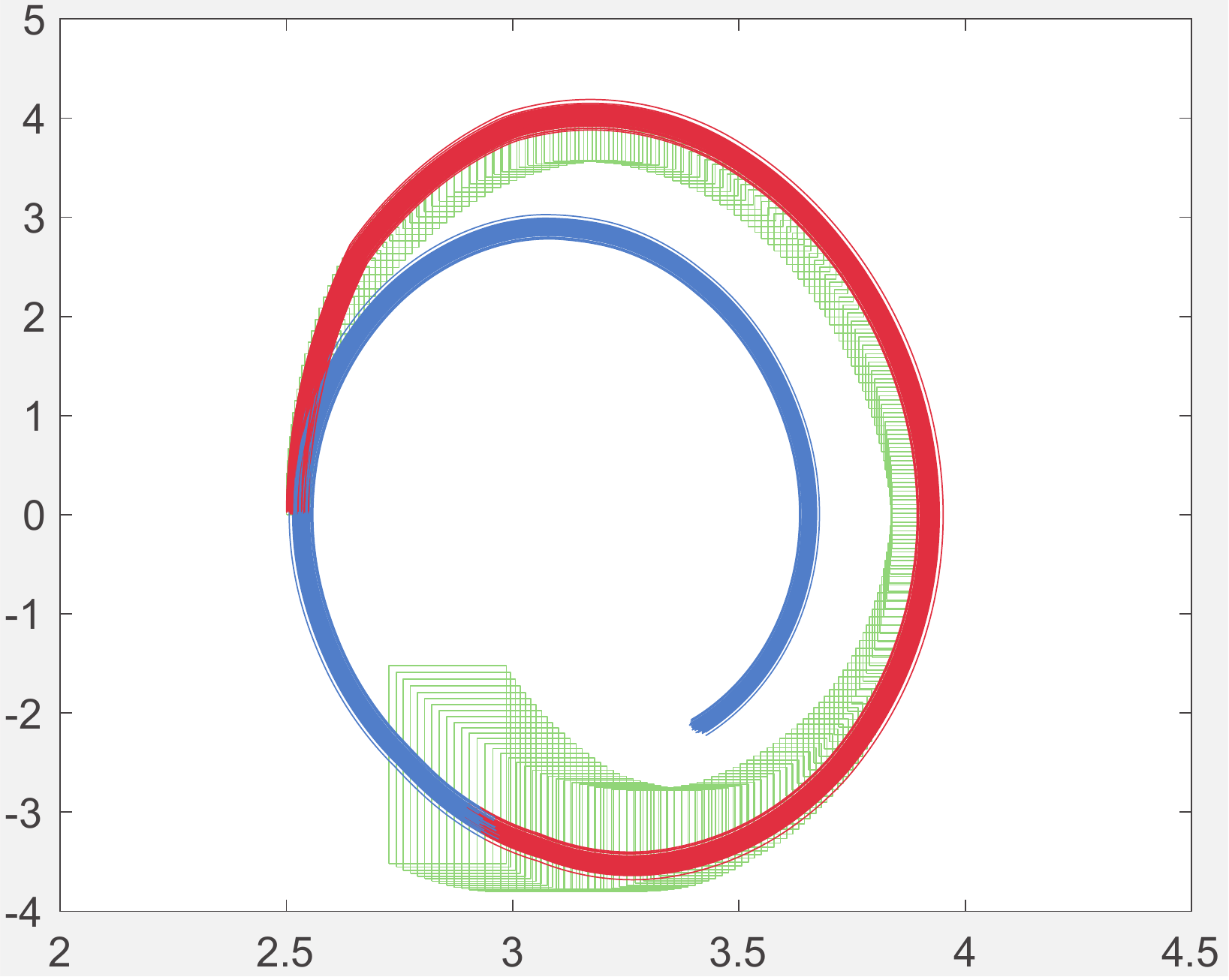}%
		\label{fig:ex10-tanh-interval}%
	}\ 
	\subfloat[][Small Lipschitz constant]{%
		\includegraphics[width=0.23\textwidth]{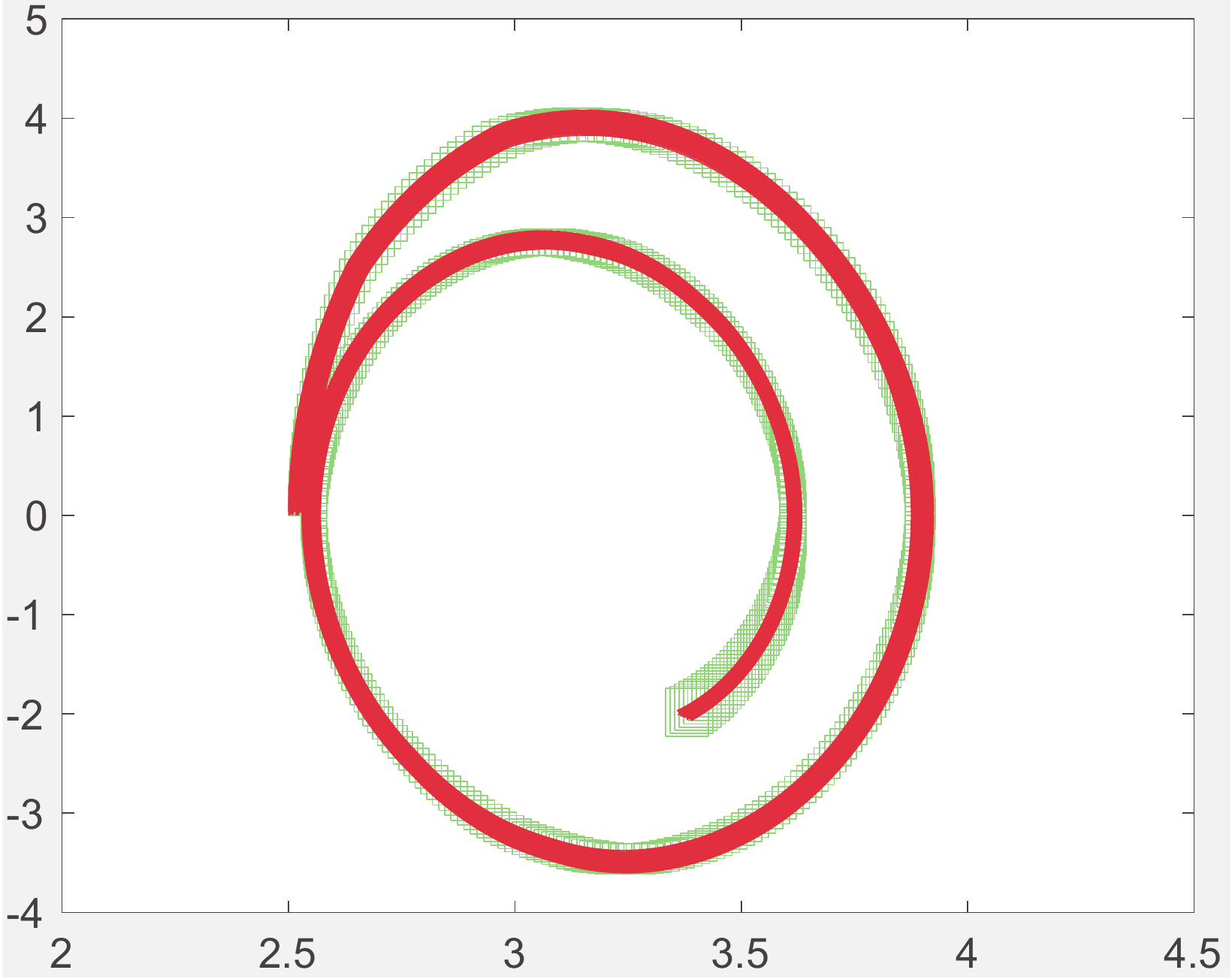}%
		\label{fig:ex10-tanh}%
	}
	\caption{ReachNN comparison between large Lipschitz constant neural-network controller and the retrained neural-network controller using Bernstein polynomials of the same degree. In Figure~\ref{fig:ex10-stability}, $x$-axis and $y$-axis are the pendulum angular position and angular velocity respectively. In Figure~\ref{fig:ex10-tanh-interval} and~\ref{fig:ex10-tanh}, $x$-axis and $y$-axis are the cart position and velocity respectively.}
	\label{fig:lipschitz}
\end{figure}

\smallskip
\noindent
\textbf{Low input dimension v.s. High input dimension.}
Thanks to the universal approximation property of Bernstein polynomials, our approach can theoretically approximate any neural network well, if the degree is high enough. However from the the perspective of implementation, we can see that the total order of the generated approximation polynomial $P_{\kappa,d}$ by ReachNN will increase exponentially along with the input dimension. Thus, our current approach may not be efficient enough to handle high dimension inputs in practice. We will investigate methods to address those high-dimensional cases in the future.


%% file: sections/conclusion.tex
\section{Conclusion} \label{sec:conclusion}

In this paper, we address the reachability analysis of neural-network controlled systems, and present a novel approach ReachNN. 
Given an input space and a degree bound, our approach constructs a polynomial approximation for a neural-network controller based on Bernstein polynomials and provides two techniques to estimate the approximation error bound. Then, leveraging the off-the-shelf tool Flow*, our approach can iteratively compute flowpipes as over-approximate reachable sets of the neural-network controlled system. The experiment results show that our approach can effectively address various neural-network controlled systems.
Our future work includes further tightening the approximation error bound estimation and better addressing high-dimensional cases.

%% file: sections/appendix.tex
\section{Appendix}
We present the plots of flowpipes for each benchmark (see Figure \ref{fig:appendix_simulation}).

\begin{figure*}[htb]
\captionsetup[subfloat]{farskip=2pt,captionskip=1pt}
	\centering	
	\subfloat[][ex1-relu]{%
		\includegraphics[width=0.24\textwidth]{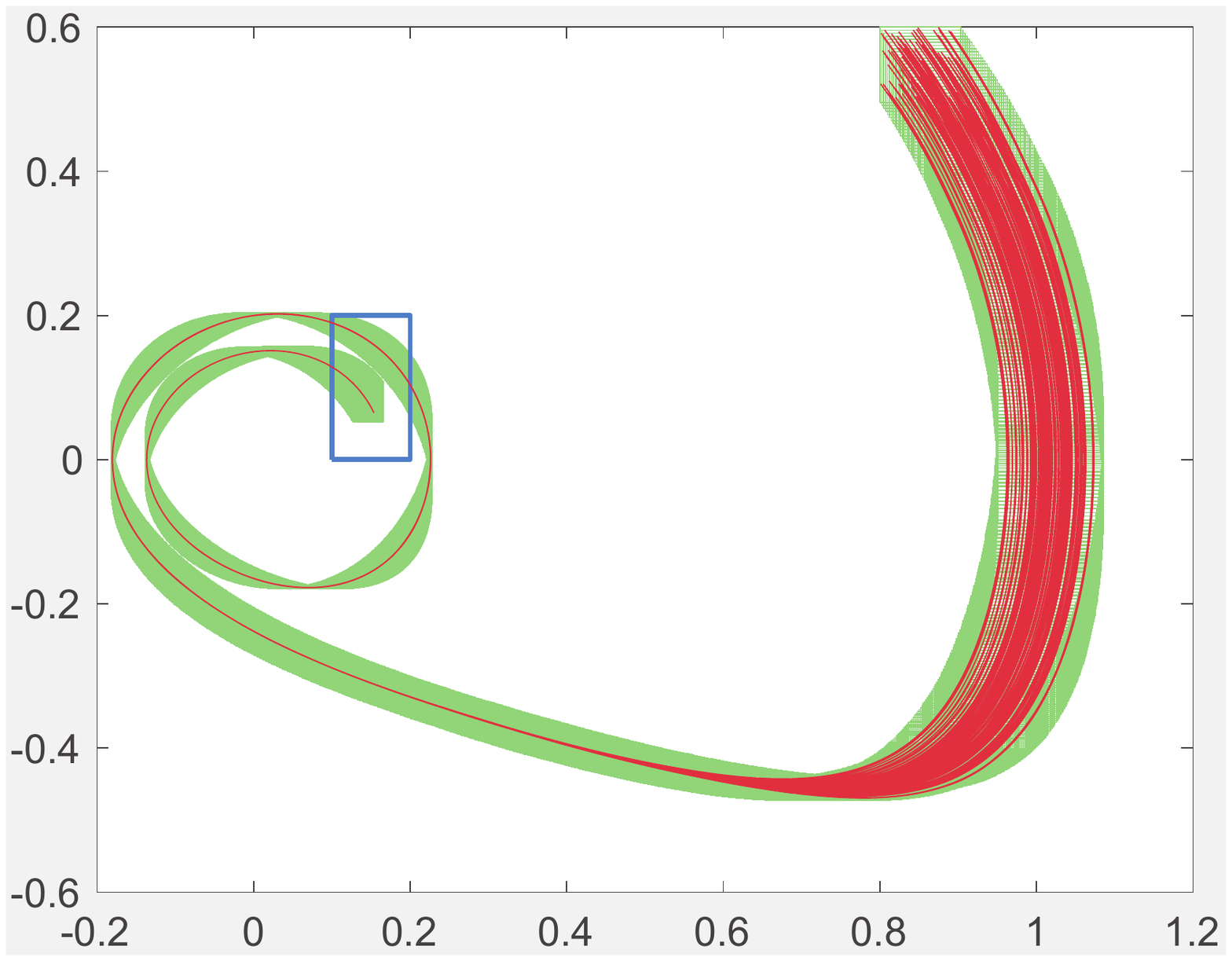}%
		\label{fig:ex1-relu}%
	}\ \
	\subfloat[][ex1-sigmoid]{%
		\includegraphics[width=0.24\textwidth]{figs/benchmark2_sigmoid.pdf}%
		\label{fig:ex1-sigmoid}%
	}\ \
	\subfloat[][ex1-tanh]{%
		\includegraphics[width=0.24\textwidth]{figs/benchmark2_tanh.pdf}%
		\label{fig:ex1-tanh}%
	}\ \
	\subfloat[][ex1-relu-tanh]{%
		\includegraphics[width=0.24\textwidth]{figs/benchmark2_relu_tanh.pdf}%
		\label{fig:ex1-relu-tanh}%
	}\\
	\subfloat[][ex2-relu]{%
		\includegraphics[width=0.24\textwidth]{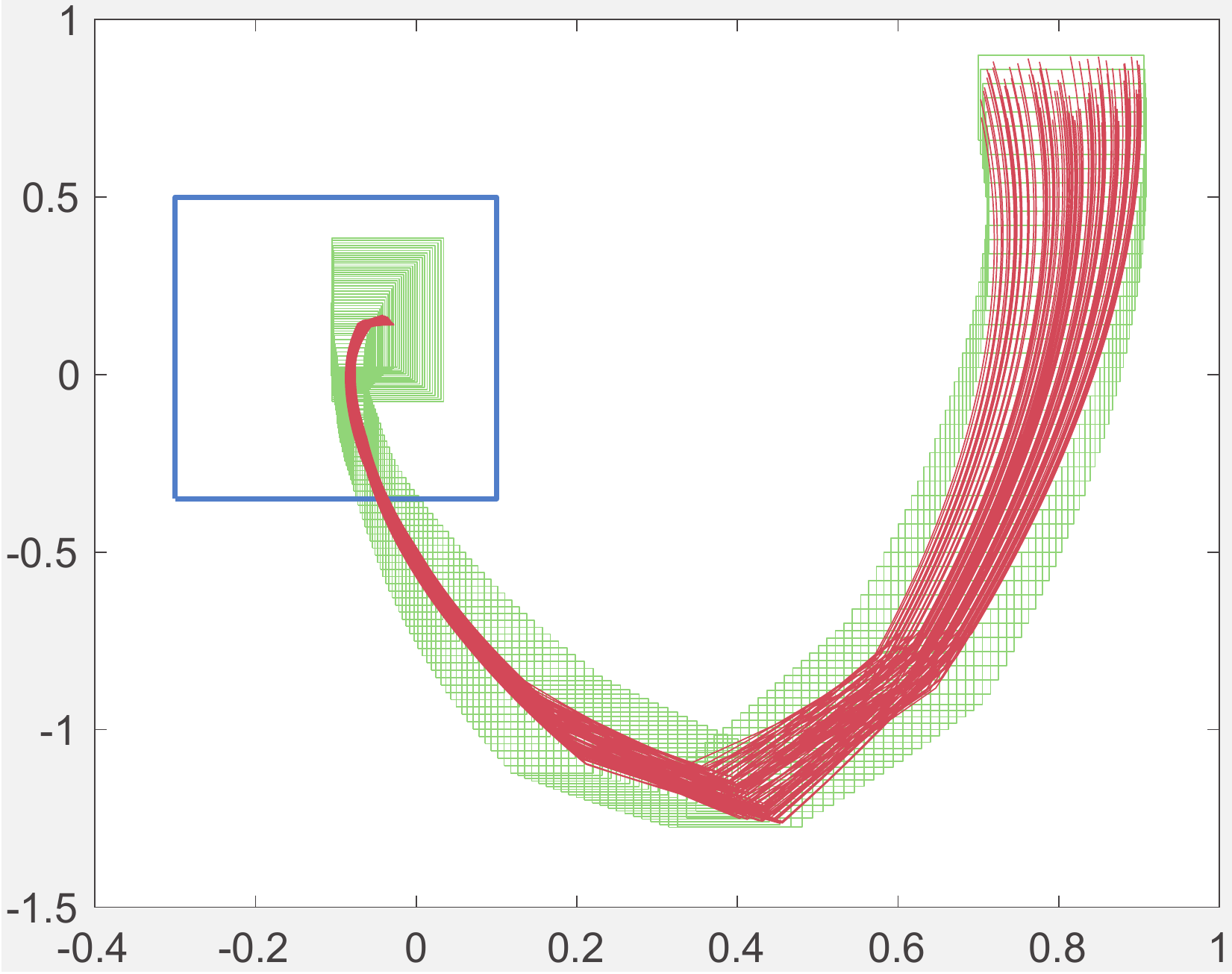}%
		\label{fig:ex2-relu}%
	}\ \
	\subfloat[][ex2-sigmoid]{%
		\includegraphics[width=0.24\textwidth]{figs/benchmark1_sigmoid.pdf}%
		\label{fig:ex2-sigmoid}%
	}\ \
	\subfloat[][ex2-tanh]{%
		\includegraphics[width=0.24\textwidth]{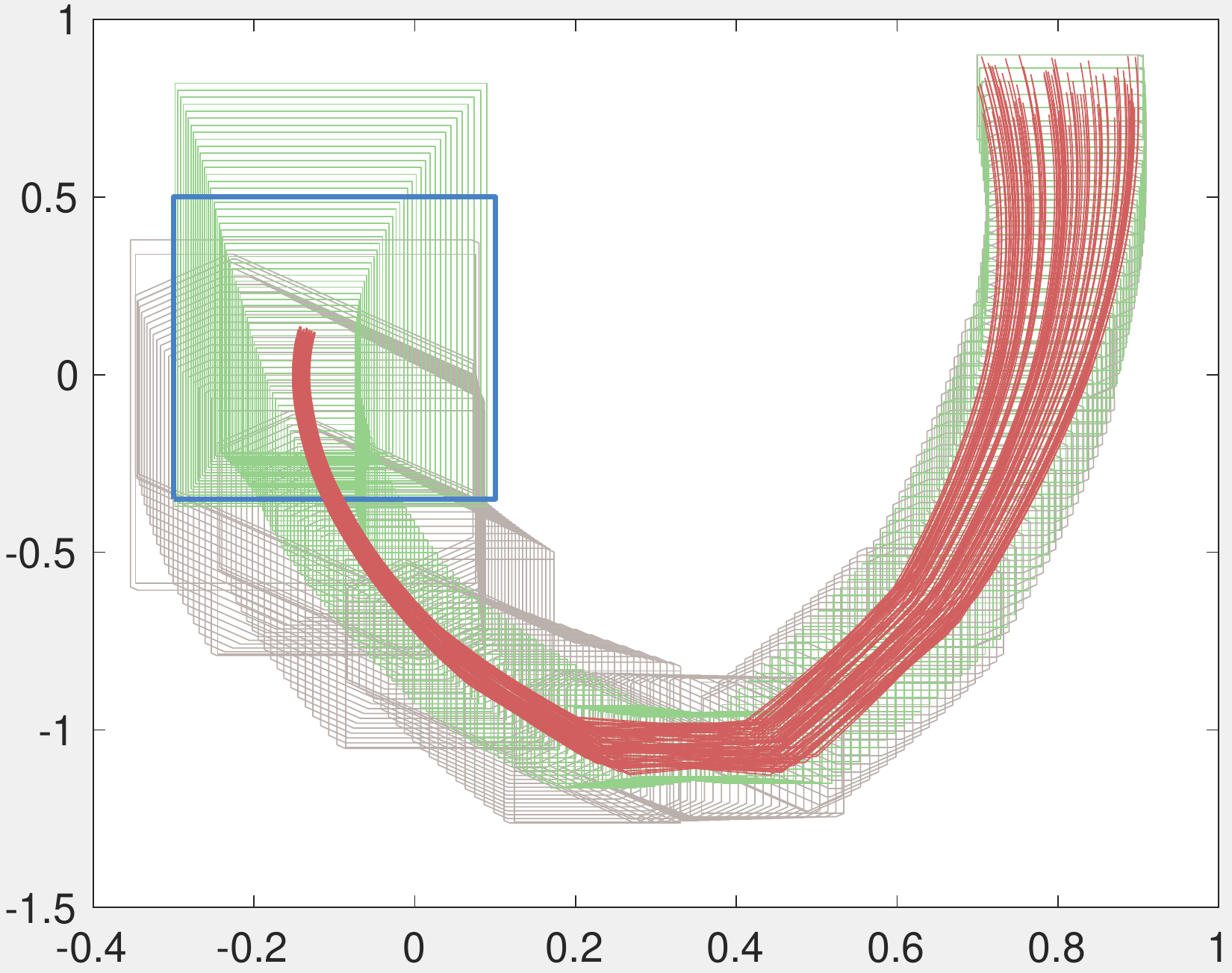}%
		\label{fig:ex2-tanh}%
	}\ \
	\subfloat[][ex2-relu-tanh]{%
		\includegraphics[width=0.24\textwidth]{figs/benchmark1_relu_tanh.pdf}%
		\label{fig:ex2-relu-tanh}%
	}\\
	\subfloat[][ex3-relu]{%
		\includegraphics[width=0.24\textwidth]{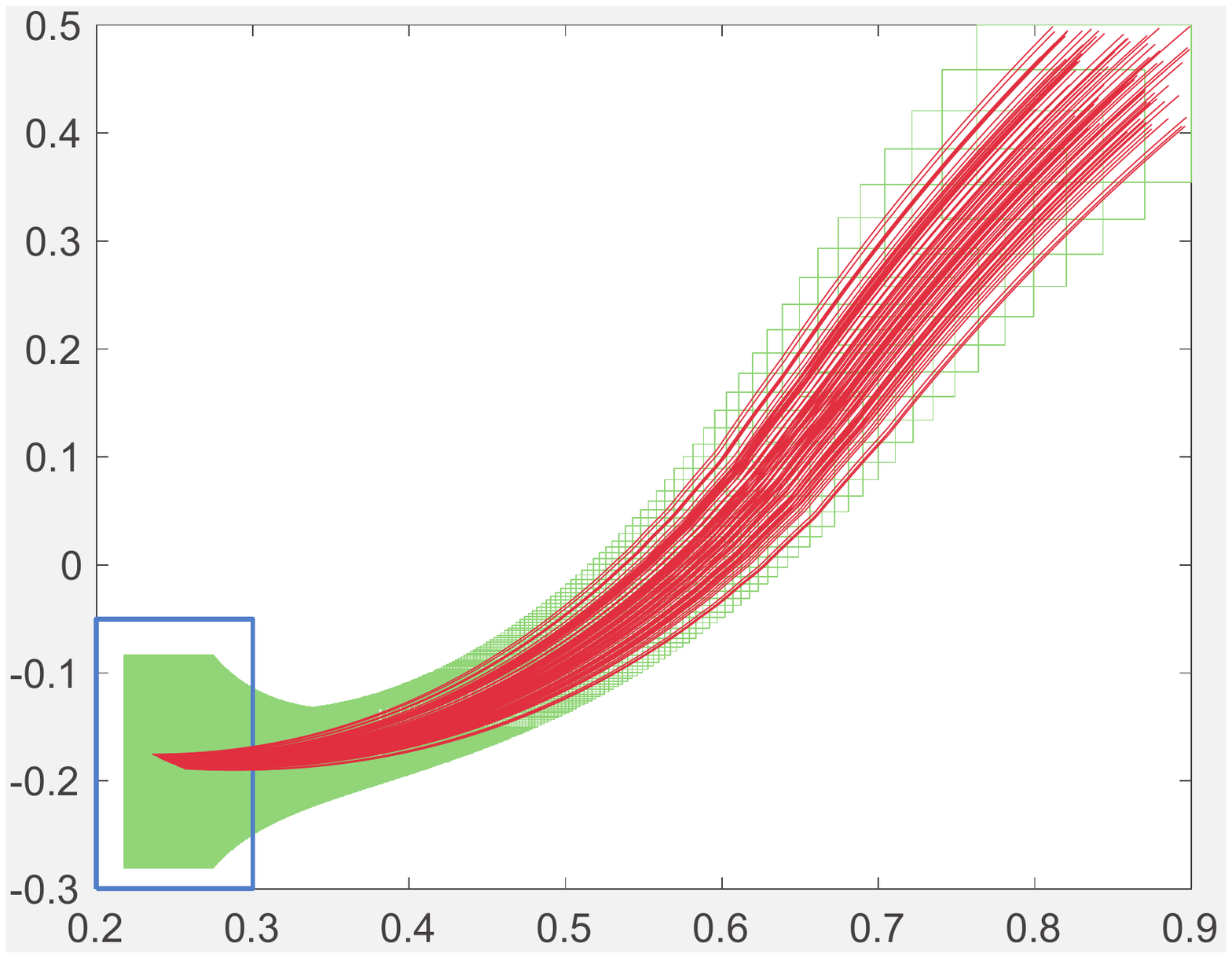}%
		\label{fig:ex3-relu}%
	}\ \
	\subfloat[][ex3-sigmoid]{%
		\includegraphics[width=0.24\textwidth]{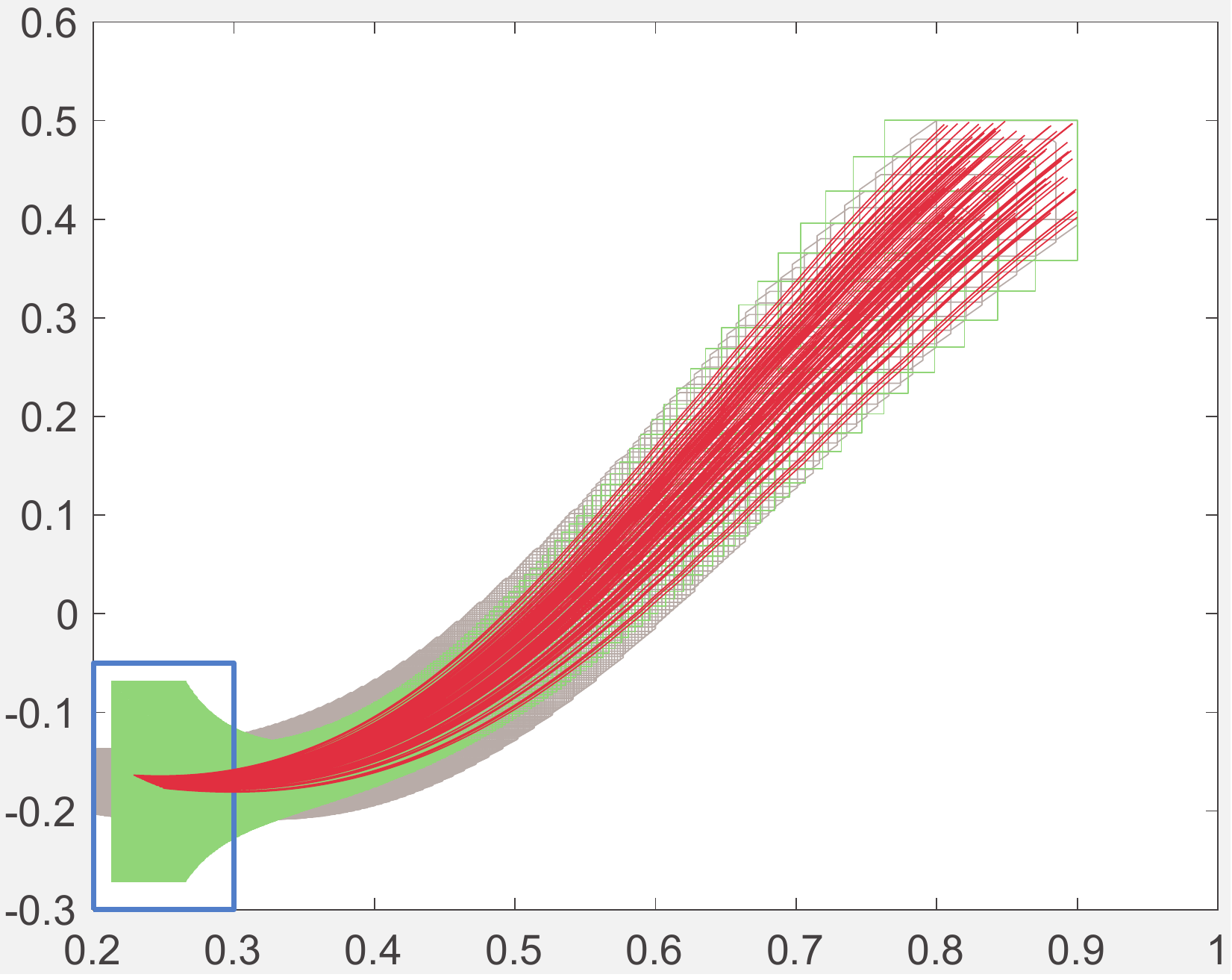}%
		\label{fig:ex3-sigmoid}%
	}\ \
	\subfloat[][ex3-tanh]{%
		\includegraphics[width=0.24\textwidth]{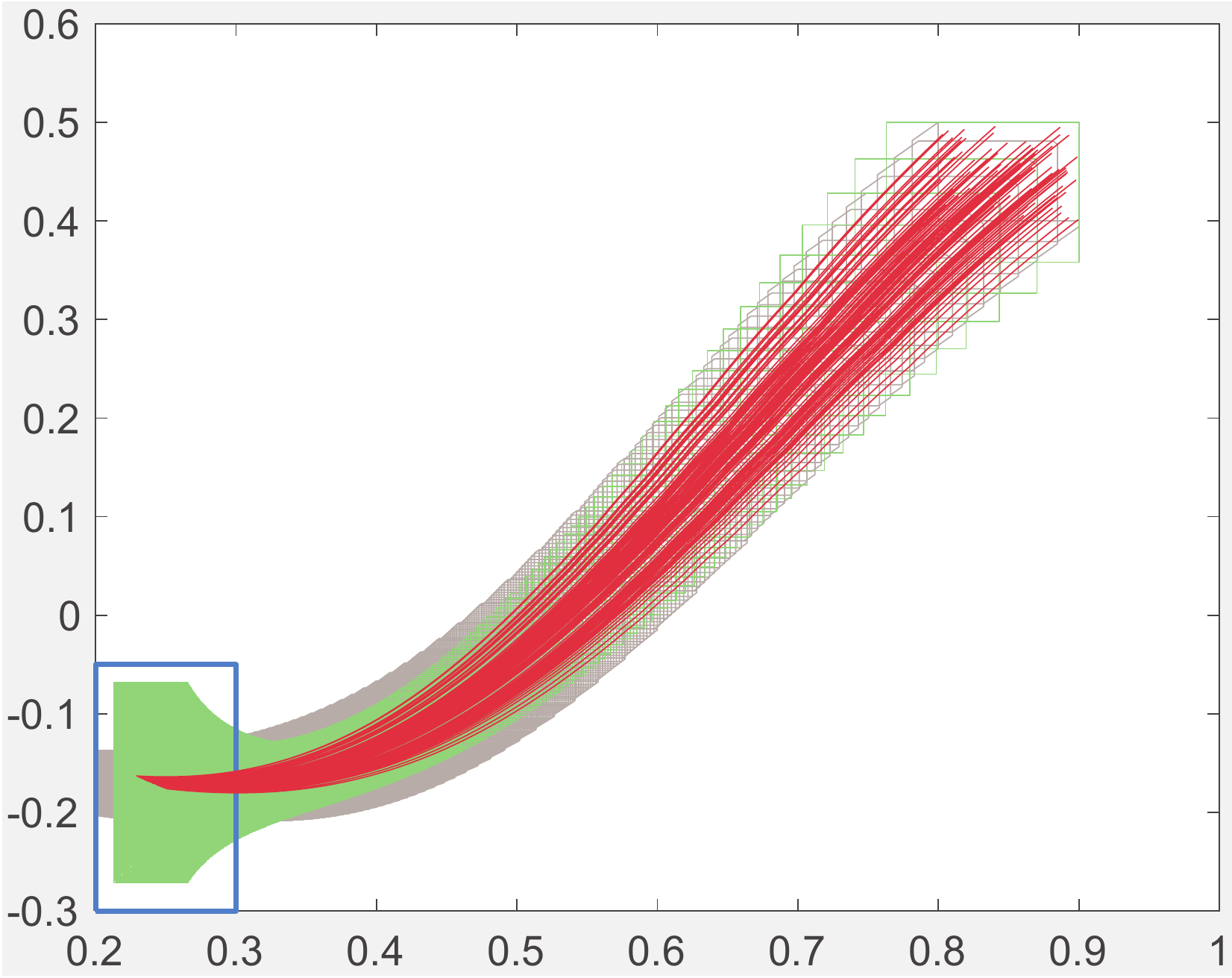}%
		\label{fig:ex3-tanh}%
	}\ \
	\subfloat[][ex3-relu-sigmoid]{%
		\includegraphics[width=0.24\textwidth]{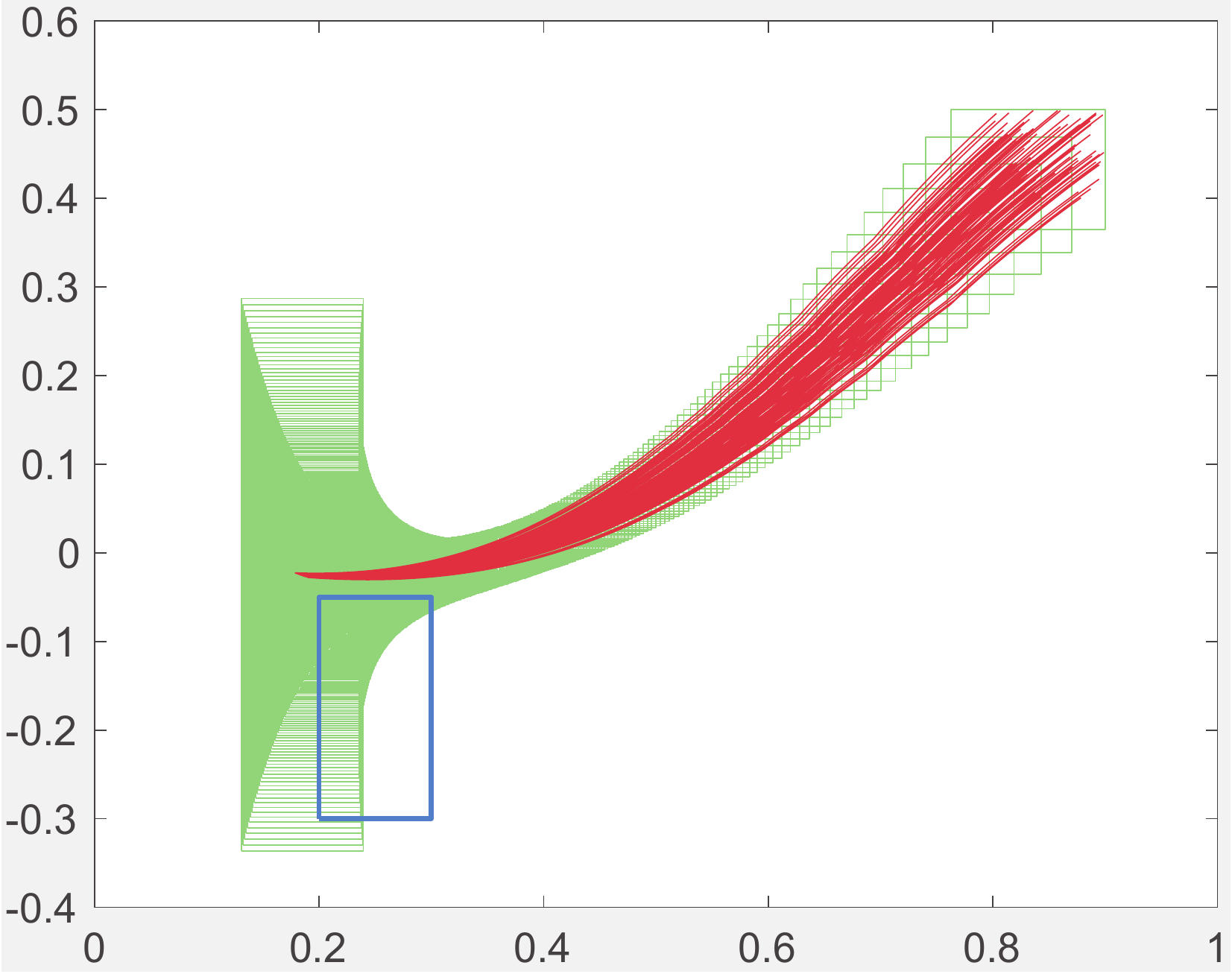}%
		\label{fig:ex3-relu-sigmoid}%
	}\\
	\subfloat[][ex4-relu]{%
		\includegraphics[width=0.24\textwidth]{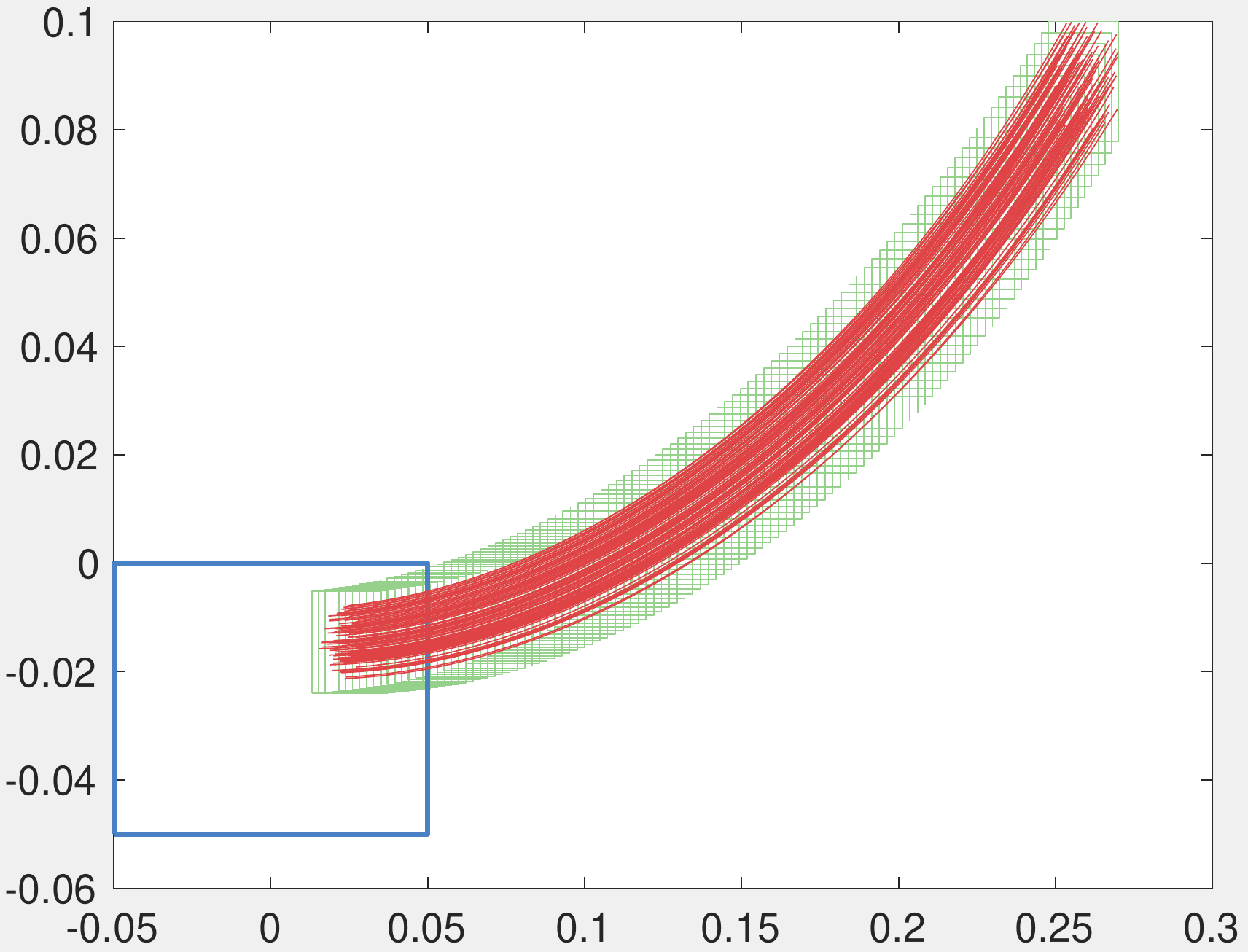}%
		\label{fig:ex4-relu}%
	}\ \
	\subfloat[][ex4-sigmoid]{%
		\includegraphics[width=0.24\textwidth]{figs/benchmark5_sigmoid.pdf}%
		\label{fig:ex4-sigmoid}%
	}\ \
	\subfloat[][ex4-tanh]{%
		\includegraphics[width=0.24\textwidth]{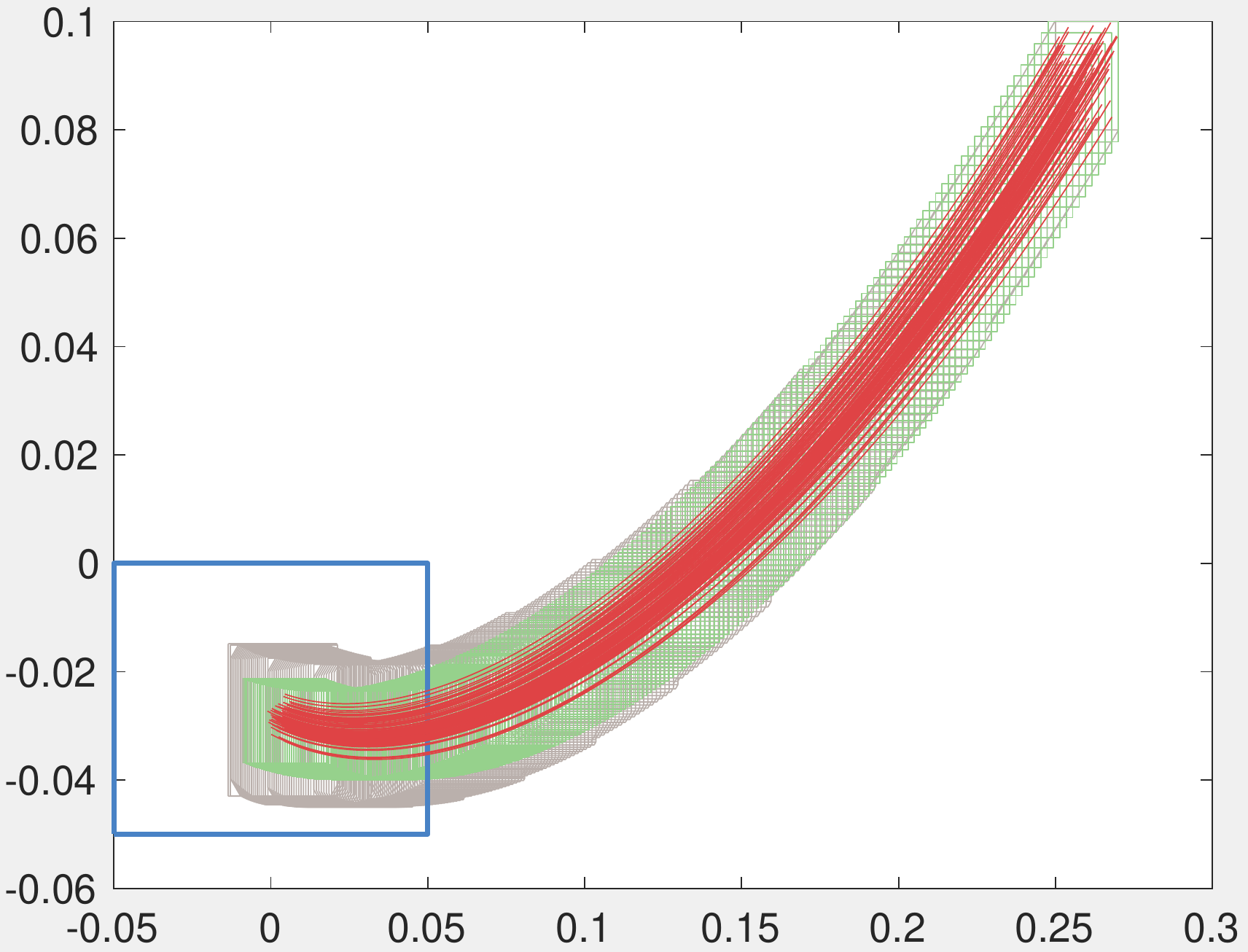}%
		\label{fig:ex4-tanh}%
	}\ \
	\subfloat[][ex4-relu-tanh]{%
		\includegraphics[width=0.24\textwidth]{figs/benchmark5_relu_tanh.pdf}%
		\label{fig:ex4-relu-tanh}%
	}\\
	\subfloat[][ex5-relu]{%
		\includegraphics[width=0.24\textwidth]{figs/benchmark7_relu.pdf}%
		\label{fig:ex5-relu}%
	}\ \
	\subfloat[][ex5-sigmoid]{%
		\includegraphics[width=0.24\textwidth]{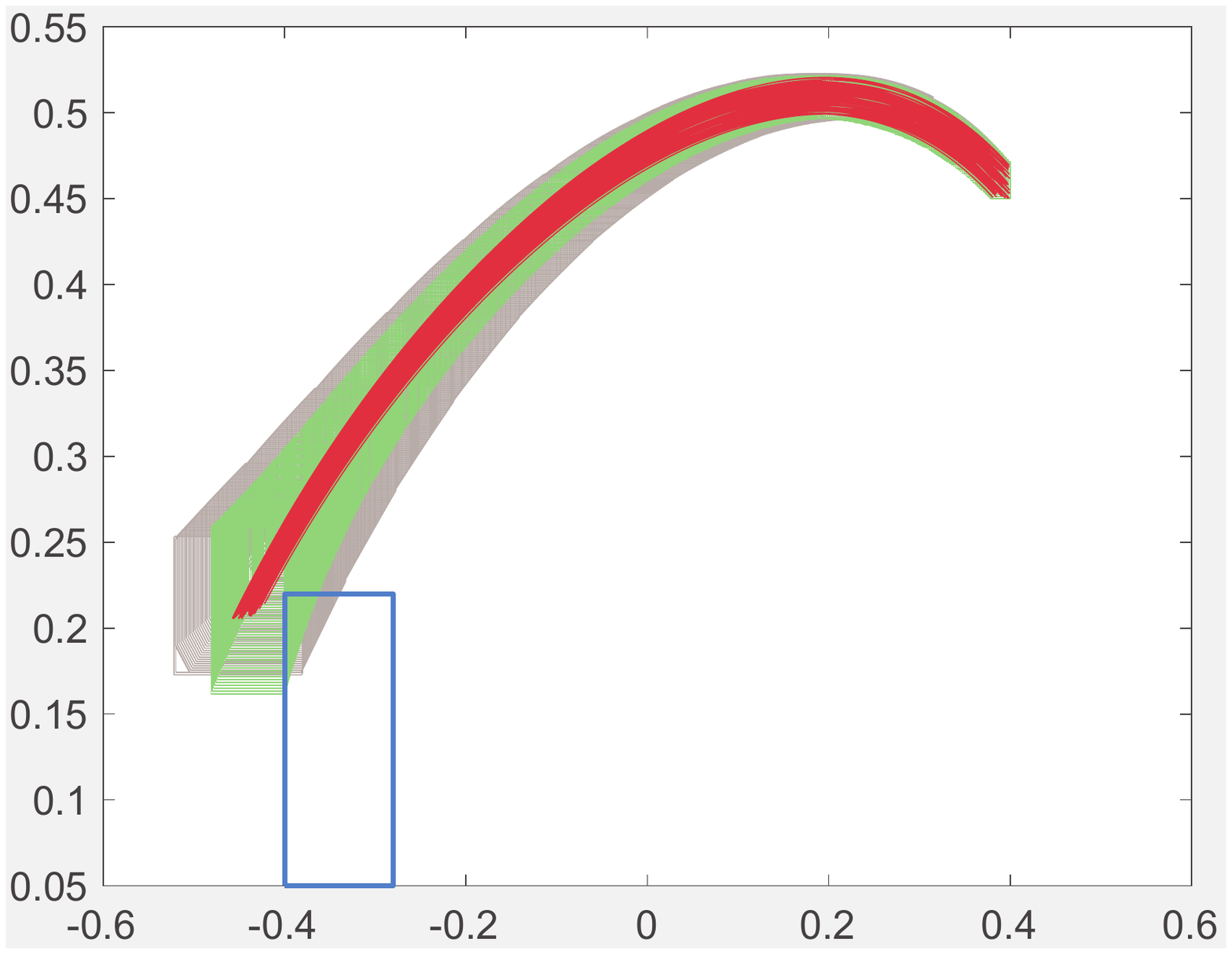}%
		\label{fig:ex5-sigmoid}%
	}\ \
	\subfloat[][ex5-tanh]{%
		\includegraphics[width=0.24\textwidth]{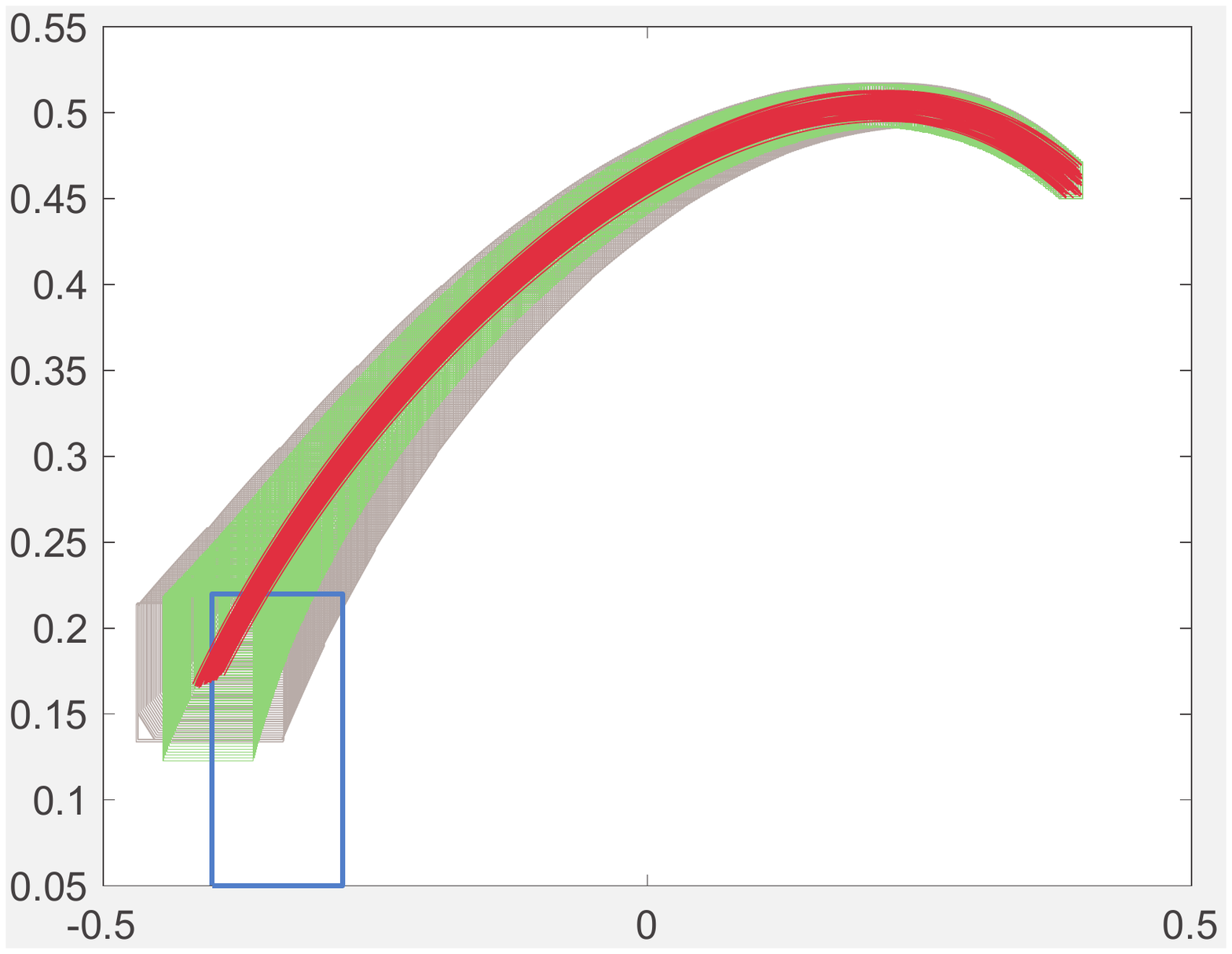}%
		\label{fig:ex5-tanh}%
	}\ \
	\subfloat[][ex5-relu-tanh]{%
		\includegraphics[width=0.24\textwidth]{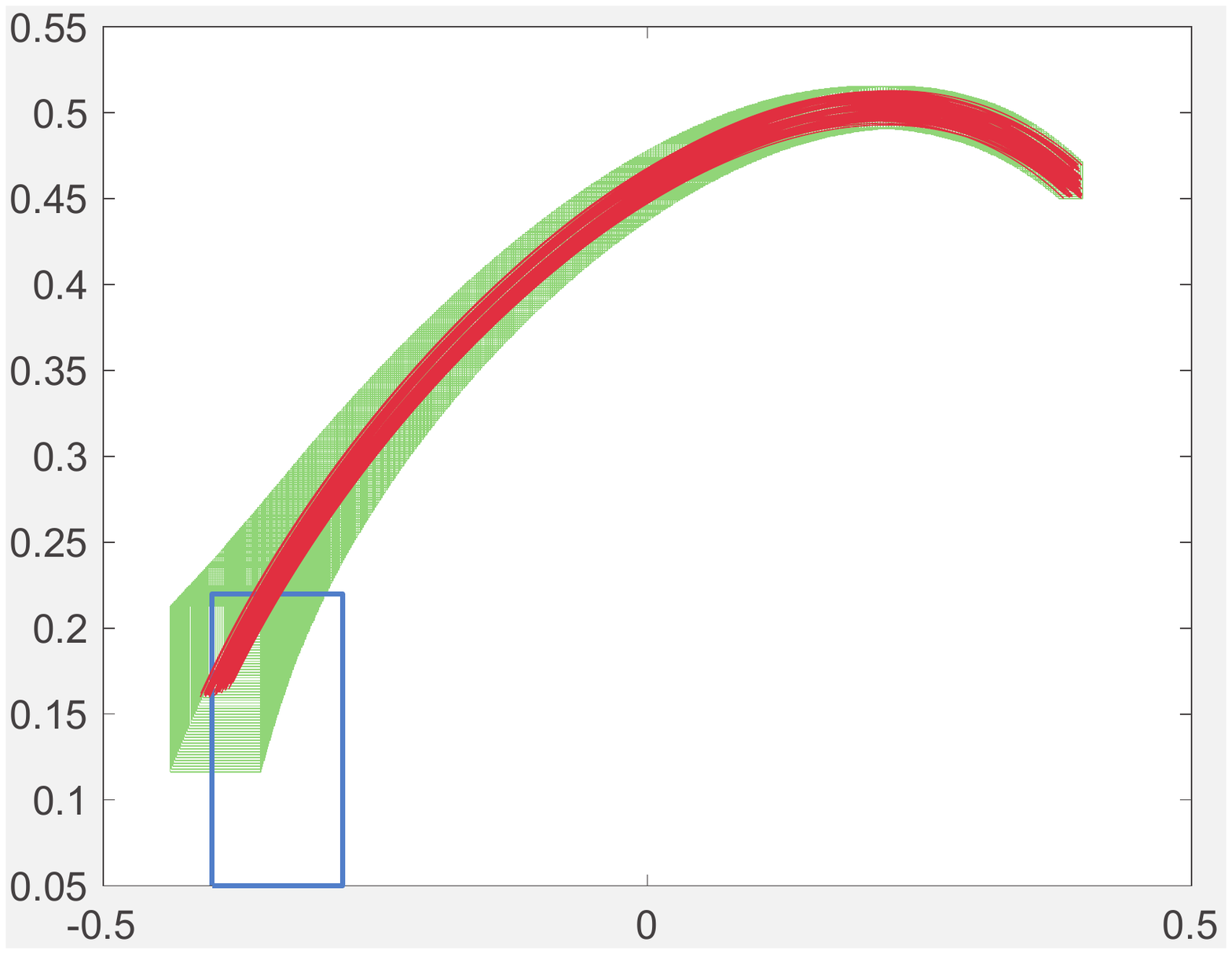}%
		\label{fig:ex5-relu-tanh}%
	}\\
	\subfloat[][ex6-relu]{%
		\includegraphics[width=0.24\textwidth]{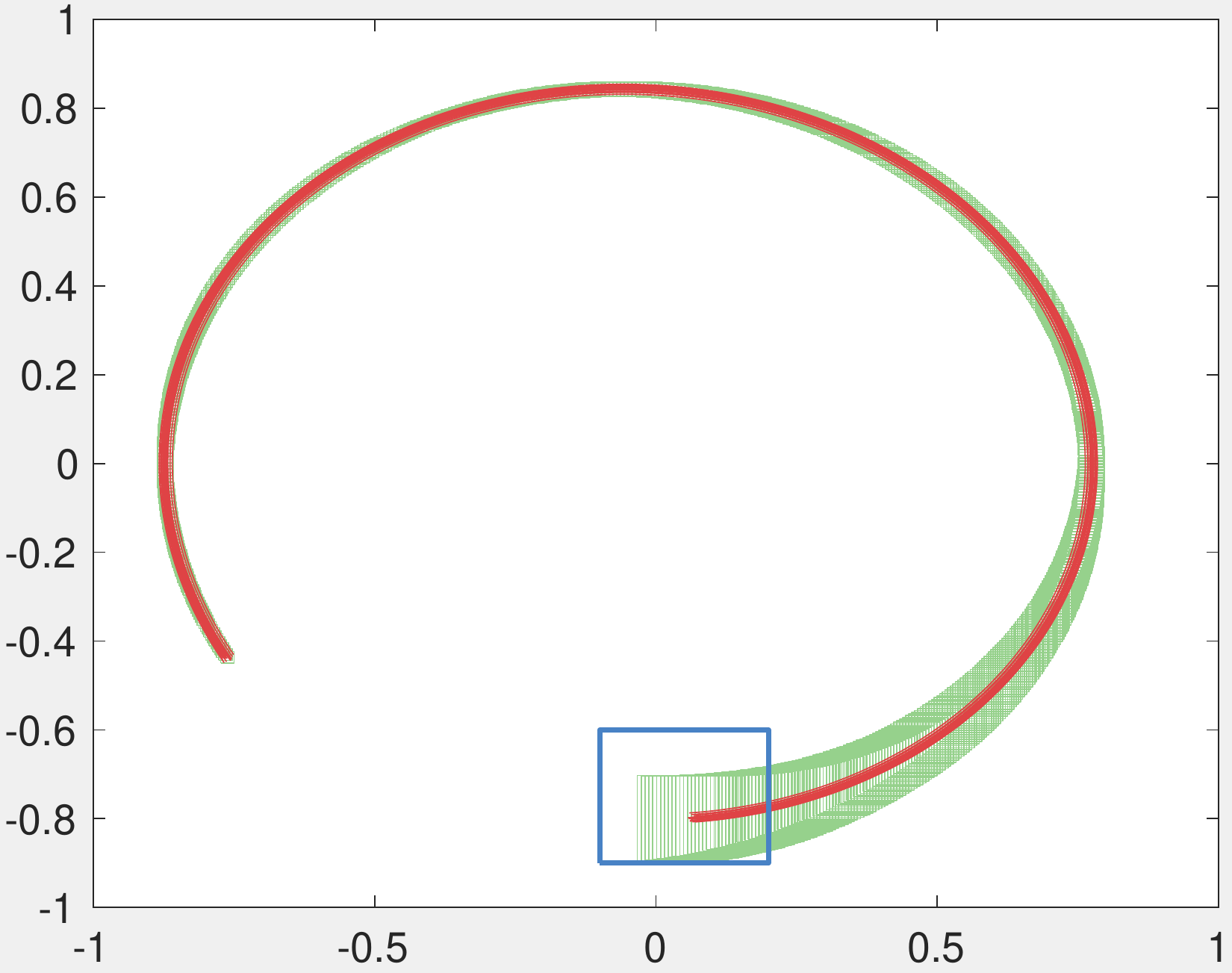}%
		\label{fig:ex6-relu}%
	}\ \
	\subfloat[][ex6-sigmoid]{%
		\includegraphics[width=0.24\textwidth]{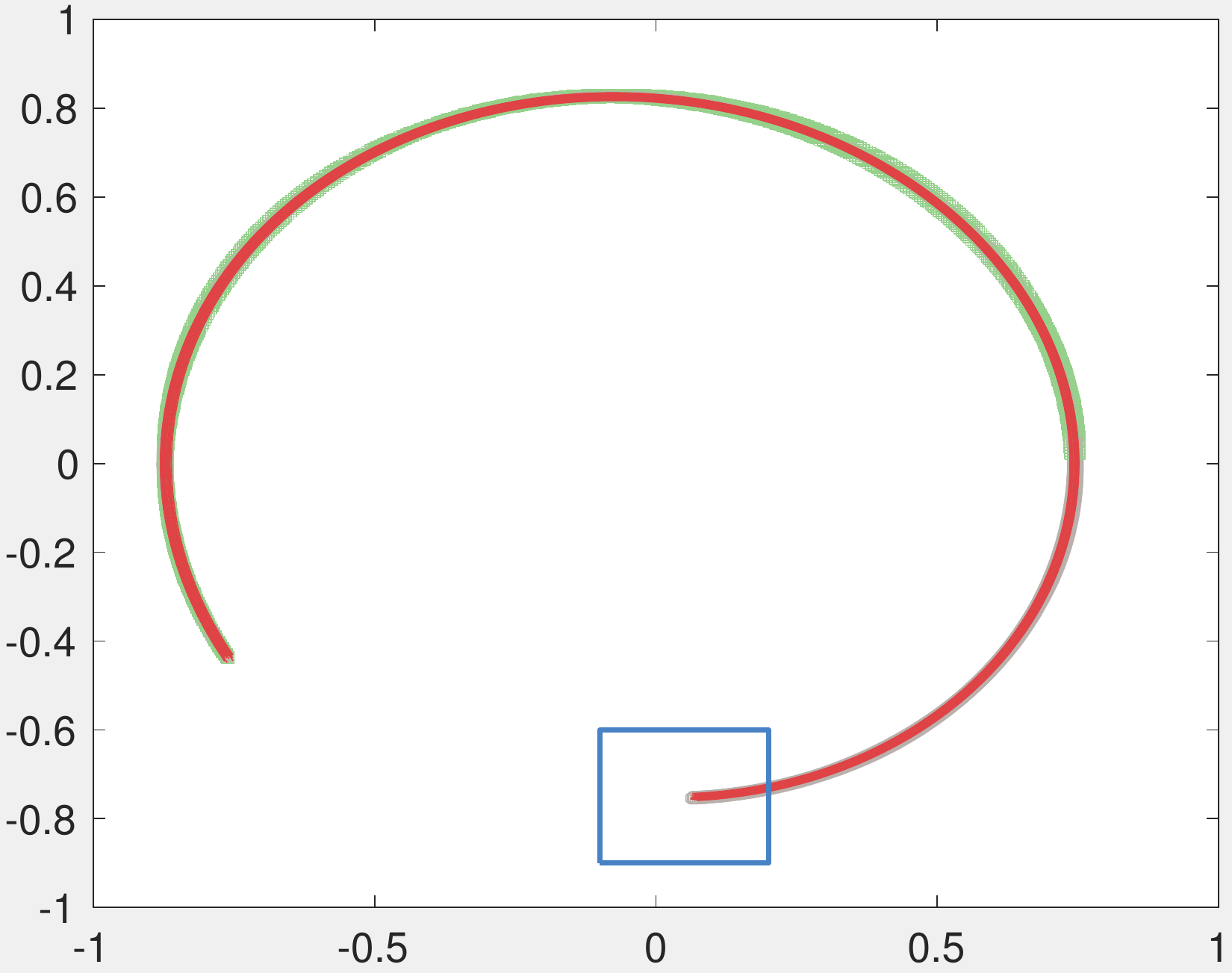}%
		\label{fig:ex6-sigmoid}%
	}\ \
	\subfloat[][ex6-tanh]{%
		\includegraphics[width=0.24\textwidth]{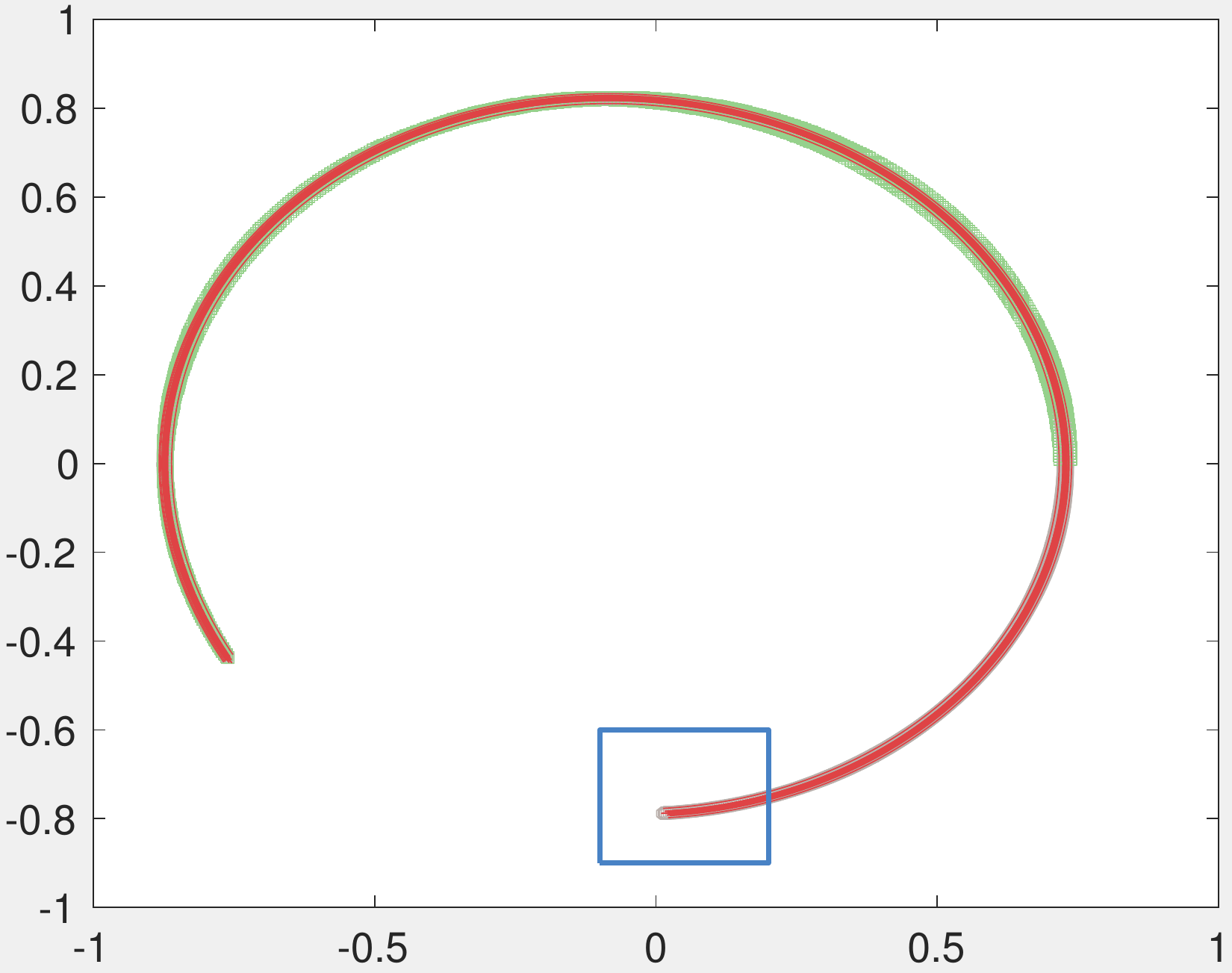}%
		\label{fig:ex6-tanh}%
	}\ \
	\subfloat[][ex6-relu-tanh]{%
		\includegraphics[width=0.24\textwidth]{figs/benchmark9_relu_tanh.pdf}%
		\label{fig:ex6-relu-tanh}%
	}
	\caption{Examples}
	\label{fig:appendix_simulation}
\end{figure*}